\documentclass[11pt,letterpaper]{article}

\usepackage{booktabs}

\usepackage[margin=1in]{geometry}
\usepackage{latexsym}
\usepackage{amssymb,amsmath,bm,tikz}
\usepackage[T1]{fontenc}
\usepackage{lmodern}
\usepackage{graphicx}
\usepackage{algpseudocode}
\usepackage{marvosym}
\usepackage{multirow,float}
\usepackage{caption}
\usepackage{subcaption}
\usepackage{comment}
\usepackage{centernot}
\usepackage{makecell}
\usepackage{refcount}
\usepackage{natbib}
\usepackage{color}
\usepackage[bookmarksopen=true, bookmarksnumbered=true,
pdfstartview=FitH, breaklinks=true, hidelinks]{hyperref}

\usepackage{dcolumn}
\newcolumntype{.}{D{.}{.}{-1}}
\newcolumntype{d}[1]{D{.}{.}{#1}}
\newcolumntype{C}{>{$}c<{$}}

\usepackage{theorem}
\theoremstyle{plain}
\theoremheaderfont{\normalfont\scshape}
\newtheorem{theorem}{Theorem}
\newtheorem{proposition}{Proposition}
\newtheorem{assumption}{Assumption}
\newtheorem{corollary}{Corollary}

\newtheorem{lemma}{Lemma}

\newcommand{\qed}{\hfill \ensuremath{\Box}}

\renewcommand{\d}{\mathrm{d}}

\DeclareMathOperator\Cov{Cov}

\DeclareMathOperator*{\argmin}{argmin}

\newenvironment{proof}{\vspace{1ex}\noindent{\bf Proof}\hspace{0.5em}}
	{\hfill\qed\vspace{1ex}}
\usetikzlibrary{decorations.markings}
\usetikzlibrary{decorations.pathmorphing}
\usetikzlibrary{shapes.geometric, arrows}
\usetikzlibrary{arrows,decorations.pathmorphing,backgrounds,positioning,fit,matrix}
\usetikzlibrary{shapes,decorations,arrows,calc,arrows.meta,fit,positioning}
\tikzset{auto,node distance =1 cm and 1 cm,semithick,
	state/.style ={circle, draw, minimum width = 0.7 cm},
	point/.style = {circle, draw, inner sep=0.04cm,fill,node contents={}},
	bidirected/.style={Latex-Latex,dashed},
	el/.style = {inner sep=2pt, align=left, sloped}
}
\usetikzlibrary{positioning}
\usetikzlibrary{fadings}
\usetikzlibrary{intersections}

\allowdisplaybreaks
\usepackage{rotating}

\usepackage{arydshln}
\usepackage{threeparttable}

\usepackage[compact]{titlesec}

\usepackage{algorithm}
\usepackage{algpseudocode}
\usepackage{pifont}
\usepackage{xcolor}

\hypersetup{pdftitle={Statistical Performance Guarantee for Subgroup Identification with Generic Machine Learning},pdfauthor={Michael Lingzhi Li and Kosuke Imai}}

\begin{document}

\newcommand\ud{\mathrm{d}}
\newcommand\dist{\buildrel\rm d\over\sim}
\newcommand\ind{\stackrel{\rm indep.}{\sim}}
\newcommand\iid{\stackrel{\rm i.i.d.}{\sim}}
\newcommand\logit{{\rm logit}}
\renewcommand\r{\right}
\renewcommand\l{\left}
\newcommand\cO{\mathcal{O}}
\newcommand\cY{\mathcal{Y}}
\newcommand\cZ{\mathcal{Z}}
\newcommand\E{\mathbb{E}}
\renewcommand\P{\mathbb{P}}
\newcommand\cL{\mathcal{L}}
\newcommand\V{\mathbb{V}}
\newcommand\cA{\mathcal{A}}
\newcommand\cB{\mathcal{B}}
\newcommand\cD{\mathcal{D}}
\newcommand\cE{\mathcal{E}}
\newcommand\cM{\mathcal{M}}
\newcommand\cU{\mathcal{U}}
\newcommand\cN{\mathcal{N}}
\newcommand\cT{\mathcal{T}}
\newcommand\cX{\mathcal{X}}
\newcommand\bA{\mathbf{A}}
\newcommand\bH{\bm{H}}
\newcommand\bB{\bm{B}}
\newcommand\bP{\bm{P}}
\newcommand\bQ{\bm{Q}}
\newcommand\bT{\bm{T}}
\newcommand\bU{\bm{U}}
\newcommand\bD{\bm{D}}
\newcommand\bS{\bm{S}}
\newcommand\bx{\bm{x}}
\newcommand\bX{\bm{X}}
\newcommand\bV{\bm{V}}
\newcommand\bW{\bm{W}}
\newcommand\bM{\bm{M}}
\newcommand\bZ{\bm{Z}}
\newcommand\bY{\bm{Y}}
\newcommand\bt{\bm{t}}
\newcommand\bbeta{\bm{\beta}}
\newcommand\bpi{\bm{\pi}}
\newcommand\bdelta{\bm{\delta}}
\newcommand\bgamma{\bm{\gamma}}
\newcommand\balpha{\bm{\alpha}}
\newcommand\bone{\mathbf{1}}
\newcommand\bzero{\mathbf{0}}
\newcommand\tomega{\tilde\omega}
\newcommand{\argmax}{\operatornamewithlimits{argmax}}

\newcommand{\R}{\textsf{R}}

\newcommand\spacingset[1]{\renewcommand{\baselinestretch}%
{#1}\small\normalsize}

\spacingset{1}

\newcommand{\tit}{\bf Statistical Performance Guarantee for Subgroup
  Identification with Generic Machine Learning}

\title{\tit\thanks{The open-source \R\ package {\sf evalITR}
      provides tools for evaluating individualized treatment rules
      (\url{https://CRAN.R-project.org/package=evalITR}).  The revised
      confidence-band formulas in this manuscript require the
      corresponding implementation updates.  We
      thank an anonymous reviewer from the Alexander and Diviya Magaro
      Peer Pre-Review Program at Harvard's Institute for Quantitative
      Social Science for helpful comments.}}

  \author{
  	Michael Lingzhi Li\thanks{Technology and Operations Management, Harvard Business School, Boston, MA
  		02163. Email: \href{mailto:mili@hbs.edu}{mili@hbs.edu}, URL:
  		\href{https://michaellz.com}{https://michaellz.com}}\hspace{.75in}
  Kosuke Imai\thanks{Department of Government and
      Department of Statistics, Harvard University, Cambridge, MA
      02138. Phone: 617--384--6778, Email:
      \href{mailto:Imai@Harvard.Edu}{Imai@Harvard.Edu}, URL:
      \href{https://imai.fas.harvard.edu}{https://imai.fas.harvard.edu}}
}

  \date{November 17, 2026}

\maketitle

\pdfbookmark[1]{Title Page}{Title Page}

\thispagestyle{empty}
\setcounter{page}{0}

\begin{abstract}
  Across a wide array of disciplines, many researchers use machine
  learning (ML) algorithms to identify a subgroup of individuals who
  are likely to benefit from a treatment the most (``exceptional
  responders'') or those who are harmed by it.  A common approach to
  this subgroup identification problem consists of two steps.  First,
  researchers estimate the conditional average treatment effect (CATE)
  using an ML algorithm.  Next, they use the estimated CATE to select
  those individuals who are predicted to be most affected by the
  treatment, either positively or negatively.  Unfortunately, CATE
  estimates are often biased and noisy.  In addition, utilizing the
  same data to both identify a subgroup and estimate its group average
  treatment effect results in a multiple testing problem.  To address
  these challenges, we develop uniform confidence bands for estimation
  of the group average treatment effect sorted by generic ML algorithm
  (GATES).  Using these uniform confidence bands, researchers can
  identify, with a statistical guarantee, a subgroup whose GATES
  exceeds a certain effect size, regardless of how this effect size is
  chosen.  The proposed methodology uses random treatment assignment,
  random sampling of units, and continuity and moment conditions,
  without requiring consistency or monotonicity of the ML score.
  We derive a Wiener representation of the estimation error, with
  corrections for the empirical subgroup threshold and the fixed
  treatment sample size.  Estimating this representation gives
  asymptotically exact uniform coverage using Gaussian draws, without
  refitting the ML model.  We also develop uniform bands for
  mean-adjusted treatment effects and other evaluation metrics.
  A simulation design and an application to a clinical trial of
  late-stage prostate cancer illustrate the subgroup identification
  problem.


\end{abstract}

\clearpage
\spacingset{1}
\section{Introduction}
\label{sec:intro}

In a diverse set of fields, machine learning (ML) algorithms are
frequently used to identify a subgroup of individuals who are likely
to benefit from a treatment the most (so called ``exceptional
responders'') or those who are negatively affected by it.  In clinical
studies, for example, identifying those who experience surprisingly
long-lasting positive health outcomes when a majority of patients do
not, is useful for understanding why a treatment works and determining
which patients should receive it
\citep[e.g.,][]{exceptional2015,10.1093/jnci/djaa061}.  Similarly, it
is critical to identify a subpopulation of patients who are harmed by
a commonly used treatment. In public policy, finding those individuals
who are likely to be helped by a treatment is a key methodological
step for improving targeting strategies
\citep[e.g.,][]{imai:ratk:13,cher:etal:19}.

A widely adopted strategy for subgroup identification proceeds in two
steps. First, an ML algorithm is used to estimate the conditional
average treatment effect (CATE) or its proxy such as biomarker
response, conditional on pre-treatment covariates. Second, individuals
are ranked by the estimated score, and those in the top (or bottom) of
the distribution are selected as exceptional responders (or those who
are most negatively affected). For example, one might define the
exceptional responders as those in the top quantile of the score and
report the average treatment effect within this group.  Such a
two-step procedure inevitably compresses rich heterogeneity into a
one-dimensional score and may therefore be suboptimal in identifying
exceptional responders. Nevertheless, it is commonly used in clinical
and policy settings due to its interpretability and practical utility
\citep[e.g.,][]{sussman2013using,razonable2022clinical,imai2023experimental}. In
the literature, the average treatment effect among individuals grouped
by the prioritization score is referred to as the Grouped Average
Treatment Effect or GATES \citep{cher:etal:19}.

As a motivating study, we consider a well-known randomized clinical
trial described by \cite{byar1980choice} concerning the effect of
administering a type of synthetic estrogen, called diethylstilbestrol,
to prolong total survival among late-stage prostate cancer patients.
The trial was overall unsuccessful as the treatment group with the
highest dosage, 5.0mg, registered an average of 0.3 month
\emph{decrease} in total survival compared to the placebo
group. However, given that the biological pathway of estrogen in
reducing prostate cancer growth is well-known, clinical researchers
believed that at least a certain subgroup of patients could have
benefited from the treatment. Therefore, many researchers have
attempted to find such a subgroup that exhibits a large GATES, by
first applying ML, optimization, or heuristic methods to create a
treatment prioritization score and then optimize over the threshold of
the score with a certain objective function
\cite[e.g.][]{bonetti2000graphical, rosenkranz2016exploratory,
  ballarini2018subgroup,bertsimas2019identifying}.

Despite its intuitive appeal and frequent use in practice, this
approach faces several methodological challenges.  First, there exists
a considerable degree of statistical uncertainty in estimating the
GATES for the selected subgroup. The treatment prioritization scores
are often based on noisy and imperfect estimates of the CATE or other
target quantities.  Second, the subgroup of exceptional responders (or
those who are negatively affected by the treatment) is usually small,
and to make the matter worse, the size of experimental data used to
identify such individuals is often not large either.  Third,
accounting for this statistical uncertainty is difficult because
utilizing the same data to both evaluate the GATES and select the
group of potential exceptional responders suffers from a multiple
testing problem. Although this problem can be alleviated by setting
aside a large holdout set for evaluation, such an approach is often
not a feasible solution given the limited sample size of randomized
control trials.

In this paper, we develop {\it uniform} confidence bands (intervals)
for estimating the GATES among a subgroup of individuals defined by a
quantile of treatment prioritization score.  This group represents the
subgroup of individuals who are predicted to benefit most from (or be
harmed most by) the treatment.  The proposed uniform confidence bands
are valid regardless of the way in which the quantile value is chosen,
thereby providing a statistical guarantee for a generic objective
function to be optimized when selecting the group of individuals who
are either most positively or most negatively affected by the
treatment.  For example, our methodology enables researchers to choose
the quantile cutoff such that the GATES among the selected exceptional
responders exceeds a certain threshold with a pre-specified
probability.

The proposed methodology is based on Neyman's repeated sampling
framework \citep{neym:23}, assuming random treatment assignment and
random sampling of individual units.  Apart from continuity and
finite second moments, our inference does not require a model for
the potential outcomes or a performance guarantee for the scoring
rule.  In particular, the ML algorithm need not estimate the CATE
consistently, and treatment effects need not be monotone in the
score.  Calibration uses Gaussian increments after a single sort of
the evaluation sample and does not refit the scoring model.

To develop our methodology (Section~\ref{sec:sample_splitting}), we
first show that the problem of finding uniform confidence bands is
equivalent to deriving the distribution of a particular functional of
induced order statistics.  This functional is also called a
concomitant of order statistics in the literature, and has been the
subject of independent interest \citep[e.g.,][]{david1974asymptotic,
  yang1977general, sen1976note}.  In particular,
\cite{bhattacharya1974convergence} and \cite{davydov2000functional}
explore the asymptotic properties of general functionals of induced
order statistics.

The limiting covariance depends on both the score ordering and the
experimental design.  We give an explicit representation using a
Wiener process and corrections for the estimated subgroup threshold
and the fixed treatment count.  Estimating the one-dimensional
functions in this representation yields a reference with the same
limiting law as the GATES estimator.  We then choose the coefficients of a Wiener-based boundary by
minimizing its area subject to simultaneous coverage under this
reference law.  The resulting band has asymptotically exact coverage.  The same representation extends to
mean-adjusted treatment effects and other evaluation metrics.

Finally, we consider a simulation design
(Section~\ref{sec:synthetic}) and the aforementioned randomized
clinical trial (Section~\ref{sec:realworld}).  The latter provides
a setting in which an unsuccessful overall treatment may nevertheless
benefit an identifiable subgroup.  The proposed methodology is
intended to distinguish such a subgroup's estimated benefit from a
benefit that is supported by a simultaneous statistical guarantee.

\paragraph{Related Literature.} The proposed methodology is related to
several strands of the existing literature in computer science,
optimization, and statistics.  First, there exists a large literature
whose goal is the development of methods to identify subgroups with
large treatment effects using ML algorithms
\citep[e.g.][]{kehl2006responder, su2009subgroup, foster2011subgroup,
  hardin2013understanding, luedtke2016optimal,
  bertsimas2019identifying, cho2021quantile, bonvini2023minimax}.
However, these existing studies either provides no statistical
performance guarantee or assumes that the threshold used for subgroup
selection is pre-specified. In contrast, we focus on developing
uniform confidence bands that are valid regardless of how such a
threshold is chosen.

Another closely related literature is the one about maximally selected
test statistics
\citep[e.g.,][]{miller1982maximally,hothorn2002maximally}.  Cast into
the setting of our paper, these authors develop a statistical test of
the null hypothesis that the CATE is independent of the treatment
prioritization score.  While this mathematical formulation is similar,
constructing uniform confidence bands must allow for arbitrary
dependence.  As discussed in Section~\ref{subsec:inference}, this
requires us to use a novel proof strategy that is different from the
ones used in the literature.

Third, \citet{cher:etal:19} considers the estimation of the GATES
based on the data from randomized experiments.  Although they use
resampling methods to provide statistical inference, the resulting
confidence bands are point-wise and are not uniformly valid across the
quantile values.  While \citet{imai:li:22} improves upon this
methodology by using Neyman's repeated sampling framework (thereby,
avoiding computationally intensive resampling methods), the authors
also only provide point-wise confidence bands.  Our work fills this
methodological gap by developing uniform inference for the GATES
estimation.  Our calibration uses a structured Gaussian reference,
related to Gaussian multiplier methods \citep{cck2013design}, with
an explicit correction for the fixed treatment count.

Finally, there exist a small but fast growing body of literature that
develops statistical methods for experimentally evaluating the
performance of data-driven individualized treatment rules
(ITRs). \cite{imai2021experimental} use Neyman's repeated sampling
framework to experimentally evaluate the efficacy of ITRs.  Building
on the up-lift modeling literature
\citep[e.g.,][]{radcliffe2007using}, the authors propose the
Population Average Prescriptive Effect (PAPE) curve as a performance
measure of an ITR relative to the non-individualized treatment rule
that treats the same proportion of randomly selected individuals.
This is a causal analogue of the Receiver Operating Characteristic
curve.  Like this paper, \cite{imai2021experimental} provide an
assumption-free statistical inference methodology to evaluate the
performance of ITRs.  Unlike our proposed methodology, however, their
confidence bands are point-wise and do not yield uniform inference for
the entire curve.

Similarly, \cite{yadlowsky2021evaluating} considers the performance
evaluation of ITRs, using a related quantity called the Targeted
Operating Characteristic (TOC) curve.  The TOC curve is similar to the
prescriptive effect curve but compares the performance of an ITR with
the treatment rule that treats everyone.  While the authors show how
to quantify statistical uncertainty for the area under the TOC curve,
their inferential methods are point-wise and are not uniformly valid
for the entire TOC curve.  Unlike these two existing works, the
proposed methodology yields uniform statistical inference, enabling
the selection of exceptional responders and other subgroups with
uniform statistical guarantees.  As shown below, the proposed uniform
confidence bands can be applied to the PAPE and TOC curves as well.

\section{The Proposed Methodology}
\label{sec:sample_splitting}

\subsection{Setup and Assumptions}
\label{subsec:randexp}

Suppose that we observe $n$ units sampled independently from a superpopulation $\mathcal{Q}$. For each unit $i=1,\ldots,n$, let $\bX_i\in\cX$ denote a vector of pre-treatment covariates and let $Y_i(1),Y_i(0)\in\cY$ denote the potential outcomes under treatment and control, respectively. We begin with the standard random sampling assumption.

\begin{assumption}[Random Sampling of Units]
	\spacingset{1}\label{asm:randomsample}
	We assume that
	\[
	\{\bX_i,Y_i(1),Y_i(0)\}_{i=1}^n \ \iid\ \mathcal{Q}.
	\]
\end{assumption}

Assumption~\ref{asm:randomsample} defines the population with respect to which we evaluate a scoring rule. In particular, all population quantities below are taken with respect to the distribution $\mathcal{Q}$ of covariates and potential outcomes.

We consider a classical randomized experiment with a binary treatment $T_i\in\{0,1\}$. Let $\bT=(T_1,\ldots,T_n)$ denote the treatment assignment vector. A total of $n_1$ units are assigned to treatment and the remaining $n_0=n-n_1$ units are assigned to control. We impose the stable unit treatment value assumption so that the two potential outcomes introduced above are well defined.

\begin{assumption}[Stable Unit Treatment Value]
	\spacingset{1}\label{asm:sutva}
	For each unit $i$, the potential outcome under an assignment vector $\bt=(t_1,\ldots,t_n)$ depends only on $t_i$. That is, for any $\bt,\bt'\in\{0,1\}^n$ such that $t_i=t_i'$,
	\[
	Y_i(\bt)=Y_i(\bt').
	\]
	We therefore write this common potential outcome as $Y_i(t_i)$. The observed outcome satisfies
	\[
	Y_i=T_iY_i(1)+(1-T_i)Y_i(0).
	\]
\end{assumption}

Assumption~\ref{asm:sutva} imposes both no interference across units and consistency between the observed and potential outcomes \citep{rubi:90}. Hence, the treatment effect for unit $i$ is represented by $Y_i(1)-Y_i(0)$ and does not depend on the assignments of the remaining units. Then, we assume the treatment is assigned by complete randomization.

\begin{assumption}[Random Assignment of Treatment]
	\spacingset{1}\label{asm:comrand}
	Conditional on $\{\bX_i,Y_i(1),Y_i(0)\}_{i=1}^n$, the treatment vector $\bT$ is uniformly distributed over
	\[
	\left\{\bt\in\{0,1\}^n:\sum_{i=1}^n t_i=n_1\right\}.
	\]
	The treatment and control sample sizes are fixed and positive, with $n_0=n-n_1$, and
	\[
	\pi_n=\frac{n_1}{n}\longrightarrow\pi\in(0,1).
	\]
\end{assumption}

Assumption~\ref{asm:comrand} implies that, conditional on the sampled covariates and potential outcomes, treatment assignment is determined solely by the complete randomization scheme. The condition $\pi_n\to\pi\in(0,1)$ ensures that both treatment arms remain asymptotically nondegenerate.

We next introduce the scoring rule $f$ used to prioritize units for treatment. Without loss of generality, we assume that a positive treatment effect is desirable and that larger scores correspond to higher treatment priority. Formally, we define $f$ as a function:
\begin{equation}
	f:\cX\longrightarrow\mathcal{S}\subset\mathbb{R}.
	\label{eq:scoring}
\end{equation}
A common choice is an estimated conditional average treatment effect,
\[
f(\bx)=\widehat{\E}\{Y_i(1)-Y_i(0)\mid\bX_i=\bx\},
\]
obtained from a machine learning algorithm. Another common choice is a baseline risk score,
\[
f(\bx)=\widehat{\E}\{Y_i(0)\mid\bX_i=\bx\}.
\]

Throughout the main analysis, we take the scoring rule $f$ as given and evaluate the treatment prioritization it induces in the target population. This includes scoring rules estimated using auxiliary data that are independent of the evaluation sample, in which case we condition on the auxiliary data and let $f$ denote the realized fitted rule. This perspective is consistent with applications in which the scoring rule is specified before evaluation, including confirmatory trials with prespecified algorithms and associated cutoffs \citep{lewis1999statistical} and evaluations of deployed machine learning systems \citep{imai2023experimental}. In Section~\ref{subsec:uncertainty}, we extend the analysis to the case in which $f$ is itself estimated and the inferential target is the performance of the underlying population scoring rule.

For each unit, let $S_i=f(\bX_i)$ denote the resulting treatment prioritization score, and let
$F(s)=\Pr(S_i\le s)$ denote its population cumulative distribution function. Because our analysis evaluates subgroups defined by the ranking induced by $S_i$, we impose the following regularity condition.

\begin{assumption}[Continuous Prioritization Score]
	\label{asm:cdf}\spacingset{1}
	The cumulative distribution function $F$ is continuous.
\end{assumption}

Assumption~\ref{asm:cdf} rules out ties in the prioritization score, so units can be ranked uniquely by $S_i$. If the scoring rule is discrete, ties can be broken using a prespecified independent continuous perturbation, which we treat as part of the scoring rule.

We next impose two regularity conditions. Let $\psi_i=Y_i(1)-Y_i(0)$ denote the individual treatment effect. Our first assumption requires the average treatment effect to vary continuously across the score distribution.

\begin{assumption}[Continuous Conditional Effect]
	\label{asm:continuity}\spacingset{1}
	The conditional average treatment effect
	\[
	p\mapsto \E\{\psi_i\mid F(S_i)=p\}
	\]
	is continuous on $[0,1]$.
\end{assumption}

Assumption~\ref{asm:continuity} allows treatment effects to vary flexibly with the prioritization score and imposes no monotonicity or sign restriction on this relationship. Similar continuity conditions appear in the literature on treatment effect heterogeneity and policy evaluation \citep[e.g.][]{wager2018estimation,yadlowsky2021evaluating}.

Our second assumption requires finite second moments of the potential outcomes.

\begin{assumption}[Moment Condition for the Potential Outcomes]
	\label{asm:moments}\spacingset{1}
	\[
	\E\{Y_i(1)^2+Y_i(0)^2\}<\infty.
	\]
\end{assumption}

Assumption~\ref{asm:moments} is sufficient for the second moments required in the asymptotic analysis below. It also implies $\E(\psi_i^2)<\infty$, while allowing the two potential outcomes to share common variation that would not be captured by a moment condition imposed only on $\psi_i$.

\subsection{Subgroup Identification with a Statistical Performance Guarantee}

We consider the setting where a researcher wishes to use a treatment prioritization score to identify a subgroup of individuals who benefit most from treatment. The same methodology can be applied to identify individuals who are harmed most by treatment by reversing the prioritization score and the sign of the treatment effect.

Formally, let $\mathcal{P}\subseteq(0,1]$ denote the set of candidate subgroup proportions. We seek the proportion $p^*$ for which the group consisting of the highest ranked $p$ proportion of the population has the largest group average treatment effect (GATES),
\begin{equation}
	p^*=\argmax_{p\in\mathcal{P}}\Psi(p),\qquad
	\Psi(p)=\E[\psi_i\mid F(S_i)\geq 1-p].
	\label{eq:pop-problem}
\end{equation}
Under Assumption~\ref{asm:comrand}, the GATES is identified from the randomized experiment as
\begin{equation}
	\Psi(p)
	=
	\E[Y_i\mid T_i=1,F(S_i)\geq1-p]
	-
	\E[Y_i\mid T_i=0,F(S_i)\geq1-p].
\end{equation}
Thus, if $\Psi(p)$ were known, Equation~\eqref{eq:pop-problem} would identify the proportion of highest ranked individuals whose average treatment effect is greatest. In practice, $\Psi(p)$ must be estimated from a finite evaluation sample. Researchers usually choose to solve the empirical analogue of Equation~\eqref{eq:pop-problem},
\begin{equation}
	\widehat p_n=\argmax_{p\in\mathcal{P}_n}\widehat\Psi_n(p),
	\label{eq:naive-opt}
\end{equation}
where $\mathcal{P}_n$ is a nonempty collection of candidate rank fractions and
\begin{equation}
	\widehat\Psi_n(p)
	=
	\frac{1}{np}\sum_{i=1}^{\lfloor np\rfloor}\widehat\psi_{[n,i]},
	\qquad
	\widehat\psi_{[n,i]}
	=
	\frac{T_{[n,i]}Y_{[n,i]}}{n_1/n}
	-
	\frac{(1-T_{[n,i]})Y_{[n,i]}}{n_0/n}.
	\label{eq:sortedest}
\end{equation}
Here, the subscript $[n,i]$ denotes the observation with the $i$th highest treatment prioritization score, so that
\[
S_{[n,1]}\geq S_{[n,2]}\geq\cdots\geq S_{[n,n]}.
\]
For example, $\mathcal{P}_n$ may contain the observed rank fractions $\{1/n,\ldots,1\}$ within a prespecified range of subgroup sizes.

While Equation~\eqref{eq:naive-opt} is straightforward to compute, it does not account for estimation uncertainty in $\widehat\Psi_n(p)$. This issue is particularly important when the evaluation sample is limited or when the subgroup of individuals who benefit most from treatment is small. Moreover, because the same outcomes are used both to choose $\widehat p_n$ and to estimate $\Psi(\widehat p_n)$, a pointwise confidence interval constructed after selecting $\widehat p_n$ does not generally provide the desired coverage guarantee.

To address this problem, we construct an asymptotic uniform one-sided confidence band $C_n(p,\alpha)$ satisfying
\begin{equation}
	\liminf_{n\to\infty}
	\P\left(
	\Psi(p)\geq\widehat\Psi_n(p)-C_n(p,\alpha)
	\text{ for all }p\in(0,1]
	\right)
	\geq1-\alpha.
	\label{eq:max_lower_conf}
\end{equation}
The uniform coverage in Equation~\eqref{eq:max_lower_conf} allows the subgroup proportion to be chosen using the same evaluation sample without invalidating the statistical guarantee. In particular, rather than selecting the proportion that maximizes the estimated GATES, we select the proportion that maximizes its lower confidence bound,
\begin{equation}
	\underline{\widehat p}_n
	=
	\argmax_{p\in\mathcal{P}_n}
	\left\{
	\widehat\Psi_n(p)-C_n(p,\alpha)
	\right\}.
	\label{eq:certified_selection}
\end{equation}
We use a deterministic rule to resolve ties. The advantage of the uniform band is that it remains valid after optimizing over the subgroup proportion. On the simultaneous coverage event in Equation~\eqref{eq:max_lower_conf},
\[
\Psi(p)\geq \widehat\Psi_n(p)-C_n(p,\alpha)
\qquad\text{for every }p\in\mathcal P_n.
\]
In particular, this inequality holds at the data dependent maximizer $\underline{\widehat p}_n$. Hence,
\begin{equation}
	\liminf_{n\to\infty}
	\P\left(
	\Psi(\underline{\widehat p}_n)
	\geq
	\widehat\Psi_n(\underline{\widehat p}_n)
	-
	C_n(\underline{\widehat p}_n,\alpha)
	\right)
	\geq1-\alpha.
	\label{eq:selected_guarantee}
\end{equation}
Equation~\eqref{eq:selected_guarantee} provides a statistical guarantee for the subgroup chosen by the procedure: its true GATES is asymptotically no smaller than the optimized lower confidence bound with probability at least $1-\alpha$. Moreover, whenever the population optimum $p^*$ is attained over the candidate domain, $\Psi(p^*)\geq \Psi(\underline{\widehat p}_n)$,
so the same optimized lower confidence bound also provides a valid lower bound on the best achievable GATES over that domain. 

The same uniform confidence band can be used for other subgroup selection objectives. For example, suppose the researcher seeks the largest subgroup whose GATES exceeds a prespecified threshold $a$. We select
\begin{equation}
	\underline{\widehat p}_n(a)
	=
	\max\{p\in\mathcal P_n:L_n(p)\geq a\}.
	\label{eq:certified_threshold}
\end{equation}
Whenever such a subgroup is selected, the uniform coverage guarantee implies that its true GATES exceeds $a$ with asymptotic probability at least $1-\alpha$. Thus, the same confidence band supports both optimization of subgroup performance and certification relative to a target effect size. Appendix~\ref{app:selection} formalizes these selection guarantees. We next construct the uniform band in Equation~\eqref{eq:max_lower_conf}.

\subsection{Constructing Uniform Confidence Bands}
\label{subsec:inference}

We now construct the uniform confidence band used for subgroup selection. Because $\widehat\Psi_n(p)$ is scaled by the subgroup proportion $p$, it is convenient to work with the cumulative GATES estimator
\[
\widehat A_n(p)=p\widehat\Psi_n(p).
\]
For $p\in[0,1]$, define
\begin{equation}
	\begin{aligned}
	h_p(s,y)&=\mathbf{1}\{F(s)\geq1-p\},\\
	A(p)&=p\Psi(p)=\E\{\psi_i h_p(S_i,Y_i(t))\},\\
	\lambda(p)&=\E\{\psi_i\mid F(S_i)=1-p\}.
	\end{aligned}
	\label{eq:cumulative_target}
\end{equation}
Since $h_p$ does not depend on its second argument, $A(p)$ is well defined for either $t\in\{0,1\}$. The quantity $\lambda(p)$ is the average treatment effect at the boundary of the subgroup containing the top $p$ proportion of units. We set $A(0)=\widehat A_n(0)=0$.

To characterize the joint estimation error of $\widehat A_n(p)$ across $p$, we introduce the Gaussian processes associated with the treatment and control samples. For each $t\in\{0,1\}$, let $\mathbb G_t$ denote a zero-mean Gaussian process indexed by measurable functions
\[
h:\mathcal S\times\mathcal Y\rightarrow\mathbb R
\]
satisfying $\E\{h(S_i,Y_i(t))^2\}<\infty$. We take $\mathbb G_1$ and $\mathbb G_0$ to be independent, with covariance
\begin{equation}
	\Cov\{\mathbb G_t(h),\mathbb G_t(g)\}
	=
	\lim_{n\to\infty}\frac{n}{n_t}
	\left[
	\E\{h(S_i,Y_i(t))g(S_i,Y_i(t))\}
	-
	\E\{h(S_i,Y_i(t))\}\E\{g(S_i,Y_i(t))\}
	\right].
	\label{eq:arm_gaussian_process}
\end{equation}

The empirical subgroup boundary is itself estimated from the ranked sample. Its contribution to the limiting process is captured by
\begin{equation}
	B_{\rm r}(p)
	=
	\lim_{n\to\infty}
	\left\{
	\frac{n_1}{n}\mathbb G_1(h_p)
	+
	\frac{n_0}{n}\mathbb G_0(h_p)
	\right\}.
	\label{eq:rank_process}
\end{equation}
The following proposition combines the estimation error in the treatment and control outcomes with the error induced by estimating this subgroup boundary.

\begin{proposition}[Functional Limit of the GATES]
	\label{prop:donsker}\spacingset{1}
	Under Assumptions~\ref{asm:randomsample}--\ref{asm:moments},
	\begin{equation}
		\sqrt n\{\widehat A_n(p)-A(p)\}
		\rightsquigarrow
		G(p)
		=
		\mathbb G_1\{y h_p(s,y)\}
		-
		\mathbb G_0\{y h_p(s,y)\}
		-
		\lambda(p)B_{\rm r}(p)
		\quad\text{in }\ell^\infty([0,1]).
		\label{eq:cumulative_fclt}
	\end{equation}
	The limiting Gaussian process $G$ has continuous sample paths.
\end{proposition}

Proposition~\ref{prop:donsker} establishes the joint asymptotic distribution of the cumulative GATES estimator over all subgroup proportions. The resulting Gaussian process, however, has a covariance structure that depends on both the ordering induced by the score and the complete randomization of treatment. To construct a uniform confidence band, we first obtain a representation of this process in terms of standard Gaussian processes.

For $t\in\{0,1\}$, define
\[
\mu_t(p)=\E\{Y_i(t)h_p(S_i,Y_i(t))\}.
\]
We then introduce two functions that characterize the covariance structure of $G$. Let
\begin{align}
	r(p)
	&=
	\lim_{n\to\infty}
	\left[
	\frac{n}{n_1}\E\{Y_i(1)^2h_p(S_i,Y_i(1))\}
	+
	\frac{n}{n_0}\E\{Y_i(0)^2h_p(S_i,Y_i(0))\}
	\right]
	-\int_0^p\lambda(u)^2\,\d u,
	\label{eq:residual_clock}\\
	d(p)
	&=
	\lim_{n\to\infty}
	\frac{n_0\mu_1(p)+n_1\mu_0(p)}
	{\sqrt{n_1n_0}}.
	\label{eq:count_covariance}
\end{align}
The function $r(p)$ is continuous and nondecreasing, and determines how the variance of the limiting process accumulates as the subgroup expands from the highest ranked units to a larger fraction of the population. The function $d(p)$ captures the component of this variation associated with fluctuations in the number of treated units. The following result uses these two functions to give a Wiener representation of the limiting process.

\begin{proposition}[Wiener Representation]
	\label{prop:representation}\spacingset{1}
	Under the conditions of Proposition~\ref{prop:donsker}, there exist a standard Wiener process $W$, a standard normal random variable $N$, and a standard Brownian bridge $B$ independent of $(W,N)$ such that
	\begin{equation}
		\Cov\{W(r(p)),N\}=d(p)
		\label{eq:wiener_count_covariance}
	\end{equation}
	and
	\begin{equation}
		G(p)\overset{d}=
		W\{r(p)\}
		+\int_0^p\{\lambda(u)-\lambda(p)\}\,\d B(u)
		-d(p)N.
		\label{eq:design_representation}
	\end{equation}
	Moreover, the process on the right hand side of Equation~\eqref{eq:design_representation} is independent of $N$.
\end{proposition}

The proof is given in Appendix~\ref{app:representation}. Proposition~\ref{prop:representation} separates the covariance structure of $G$ into three components. The Wiener process $W\{r(p)\}$ captures the variation that accumulates as progressively larger portions of the ranked population are included. The Brownian bridge term captures the additional variation from estimating the score quantile that determines the subgroup boundary. Its integrand depends on the difference between the treatment effect at each rank and the treatment effect at the subgroup boundary, and the term therefore disappears when $\lambda(p)$ is constant across the score distribution.

The final term reflects the complete randomization design. Under independent treatment assignment, the number of treated units would fluctuate across repeated samples, and this fluctuation is correlated with $W\{r(p)\}$ through $d(p)$. Complete randomization fixes the number of treated units at $n_1$. The subtraction of $d(p)N$ removes this component of the variation, yielding a process that is independent of the treatment count fluctuation represented by $N$. This representation allows us to construct the uniform confidence band from standard Gaussian processes while retaining the dependence induced by both the score ordering and the randomized design.

The representation therefore provides a direct way to approximate the distribution of $G$ by estimating $r$, $d$, and $\lambda$ from the ordered evaluation sample. We partition the ordered observations into $J_n$ consecutive blocks, where $J_n\to\infty$ and $J_n^2/n\to0$, and estimate $\lambda$ locally by averaging $\widehat\psi_{[n,i]}$ within each block. The corresponding within-block residuals provide estimators $\widehat r_n$ and $\widehat d_n$ of the cumulative residual variance and treatment-count covariance. The exact construction is given in Appendix~\ref{app:reference_consistency}. These estimators satisfy
\begin{equation}
	\sup_{p\in[0,1]}|\widehat\lambda_n(p)-\lambda(p)|
	+\sup_{p\in[0,1]}|\widehat r_n(p)-r(p)|
	+\sup_{p\in[0,1]}|\widehat d_n(p)-d(p)|
	\xrightarrow{\P}0.
	\label{eq:reference_parameter_consistency}
\end{equation}
Substituting them into the representation in Proposition~\ref{prop:representation} yields a Gaussian reference process $\widehat G_n$ that, conditional on the evaluation sample, satisfies
\begin{equation}
	\widehat G_n\rightsquigarrow_{\P}G
	\quad\text{in }\ell^\infty([0,1]).
	\label{eq:reference_process_consistency}
\end{equation}
We use this reference process to choose the uniform boundary. Let $c=A(1)$ and define
\begin{equation}
	v_c(p)
	=r(p)+\int_0^p\{\lambda(u)-c\}^2\,\d u,
	\qquad
	\widehat v_n(p)
	=\frac{1}{n}\sum_{i=1}^{\lfloor np\rfloor}
	\{\widehat\psi_{[n,i]}-\widehat A_n(1)\}^2.
	\label{eq:comparison_variance_est}
\end{equation}
We consider the boundary
\begin{equation}
	\widehat h_{\beta,n}(p)
	=\beta_0\sqrt{\widehat v_n(1)}
	+\beta_1\sqrt{\widehat v_n(p)}.
	\label{eq:direct_boundary}
\end{equation}
Under the time transformation $t=v_c(p)/v_c(1)$, this boundary corresponds to the family $\beta_0+\beta_1\sqrt{t}$ studied by \citet{kendall2007confidence}. The first term prevents the boundary from collapsing near the smallest subgroup sizes, while the second allows its width to adapt to the variation accumulated along the GATES curve.

For a fixed $\epsilon>0$, define
\begin{equation}
	\mathcal B_\epsilon
	=\{(\beta_0,\beta_1)\in[0,\infty)^2:
	\beta_0\geq\epsilon\beta_1\}.
	\label{eq:boundary_class}
\end{equation}
Let $\P^*$ denote probability with respect to the Gaussian reference process conditional on the evaluation sample. We choose the boundary with the smallest area among those attaining simultaneous coverage $1-\alpha$:
\begin{equation}
	\begin{aligned}
		(\widehat\beta_0^*,\widehat\beta_1^*)
		&=\argmin_{\beta\in\mathcal B_\epsilon}
		\int_0^1\widehat h_{\beta,n}(p)\,\d p\\
		\text{subject to}\qquad
		&\P^*\{\widehat G_n(p)\leq\widehat h_{\beta,n}(p),
		\ \forall p\in[0,1]\}\geq1-\alpha.
	\end{aligned}
	\label{eq:CIoptim}
\end{equation}
The resulting confidence band is given in the following theorem.

\begin{theorem}[Uniform Minimum-area Band for the Cumulative GATES]
	\label{thm:sorteduniform}\spacingset{1}
	\mbox{}\par
	Suppose Assumptions~\ref{asm:randomsample}--\ref{asm:moments} hold,
	$v_c(1)>0$ for $c=A(1)$, and $J_n\to\infty$ with $J_n^2/n\to0$.
	For fixed $\epsilon>0$ and $0<\alpha<1/2$, let
	$(\widehat\beta_0^*,\widehat\beta_1^*)$ solve
	Equation~\eqref{eq:CIoptim}, and define
	\begin{equation}
		C_n(p,\alpha)
		=\frac{
			\widehat\beta_0^*\sqrt{\widehat v_n(1)}
			+\widehat\beta_1^*\sqrt{\widehat v_n(p)}
		}{p\sqrt n},
		\qquad 0<p\leq1.
		\label{eq:main_uniform_band}
	\end{equation}
	Then
	\begin{equation}
		\lim_{n\to\infty}\P\left\{
		\Psi(p)\geq\widehat\Psi_n(p)-C_n(p,\alpha),
		\ \forall p\in(0,1]\right\}=1-\alpha.
		\label{eq:uniform_coverage}
	\end{equation}
	Moreover, the fitted boundary is asymptotically minimum-area within
	$\mathcal B_\epsilon$.
\end{theorem}
The proof is given in Appendix~\ref{app:sorteduniform}. Theorem~\ref{thm:sorteduniform} chooses the coefficients of the confidence band to minimize its area while retaining simultaneous coverage over the entire GATES curve. The optimization uses the full reference process in Proposition~\ref{prop:representation}, and therefore incorporates both the uncertainty from estimating the subgroup threshold and the dependence induced by complete randomization. The same construction yields an upper confidence band by reversing the sign of the reference process and a two-sided confidence band by calibrating against $|\widehat G_n|$.

The objective in Equation~\eqref{eq:CIoptim} minimizes the area of the cumulative error boundary over $p\in[0,1]$. In many applications, however, researchers may only consider subgroups larger than some prespecified fraction $p_l>0$. In this case, the coefficients can instead be chosen to minimize the area of the corresponding GATES confidence band over $p\in[p_l,1]$:
\begin{equation}
	\int_{p_l}^1\frac{\widehat h_{\beta,n}(p)}{p}\,\d p.
	\label{eq:minimump}
\end{equation}
The simultaneous coverage constraint remains unchanged over $p\in[0,1]$, so this modification affects only the allocation of width across subgroup sizes. One may similarly minimize the boundary area after the reference time transformation $t=v_c(p)/v_c(1)$. Because each criterion is linear in $(\beta_0,\beta_1)$, the coverage result in Theorem~\ref{thm:sorteduniform} continues to hold under either choice.

The family in Equation~\eqref{eq:direct_boundary} includes the constant boundary as the special case $\beta_1=0$. Allowing $\beta_1>0$ reallocates width across the GATES curve and can reduce the area of the resulting confidence region without affecting validity. The optimization remains low dimensional: for any fixed $\beta_0$, the coverage probability is monotone in $\beta_1$, so the coefficients can be obtained by a one-dimensional search over $\beta_0$ combined with a binary search over $\beta_1$. Appendix~\ref{app:calibration} gives the resulting algorithm and establishes its Monte Carlo error control.

The boundary family in Equation~\eqref{eq:direct_boundary} is not unique. We next consider a power family that allows the width of the reference boundary to vary more flexibly with the accumulated variance.

\begin{corollary}[Power Reference Boundaries]
	\label{cor:conf_band_k_family}\spacingset{1}
	Under the conditions of Theorem~\ref{thm:sorteduniform}, fix $0\leq\kappa<1/2$ and $0<\eta<1$. Let
	\begin{equation}
		\widehat t_n(p)=\frac{\widehat v_n(p)}{\widehat v_n(1)}
	\end{equation}
	and define
	\begin{equation}
		\widehat\beta_\alpha^{(\kappa,\eta)}
		=\inf\left\{b\geq0:
		\P^*\left(
		\widehat G_n(p)\leq
		b\sqrt{\widehat v_n(1)}
		\{\widehat t_n(p)\vee\eta\}^{\kappa},
		\ \forall p\in(0,1]
		\right)\geq1-\alpha
		\right\}.
		\label{eq:k_family}
	\end{equation}
	Then
	\begin{equation}
		\frac{\widehat\beta_\alpha^{(\kappa,\eta)}
			\sqrt{\widehat v_n(1)}
			\{\widehat t_n(p)\vee\eta\}^{\kappa}}
		{p\sqrt n},
		\qquad 0<p\leq1,
	\end{equation}
	defines a lower uniform confidence band with asymptotically exact simultaneous coverage $1-\alpha$. The same result holds with $\eta=0$ on any fixed interval $[p_l,1]$ satisfying $p_l>0$ and $v_c(p_l)>0$. For $\kappa=0$, the result holds on $(0,1]$ without a floor.
\end{corollary}

The proof is given in Appendix~\ref{app:conf_band_k_family}. The restriction $\kappa<1/2$ reflects the behavior of the Wiener process near the origin. By the law of the iterated logarithm, a boundary proportional to $\sqrt t$ cannot uniformly bound a Wiener process on the full interval $[0,1]$ with positive probability. The floor $\eta$ prevents the power boundary from vanishing at the smallest subgroup sizes. This regularization is unnecessary for the minimum-area band in Theorem~\ref{thm:sorteduniform}, whose positive intercept provides the corresponding protection near $p=0$.

\subsection{Numerical Illustration}
\label{subsec:numeric}

We present a numerical example to illustrate the uniform confidence
bands of different forms introduced above. We utilize the data
generating process described in Section~\ref{sec:synthetic} and train
the LASSO regression model by fitting it to the population data based
on the empirical distribution of the full sample. Then, we generate a
random sample of size $n=2,500$ from this population and apply our
methodology to the CATE estimates from the LASSO
regression. Specifically, we compute the proposed minimum-area uniform
confidence bands in two ways, based on Equation~\eqref{eq:CIoptim}
and Equation~\eqref{eq:minimump} with $p_l=0.05$, and the bands in
Equation~\eqref{eq:k_family} with
$\kappa=0$ and $\kappa=0.499$, using $\eta=0.001$ for the positive
exponent. We use $J_n=\lfloor n^{1/3}\rfloor$ blocks and
$\epsilon=10^{-3}$, and calibrate each band with 10,000 Gaussian draws.
The slope grid runs from zero to six in increments of $0.025$.
We also provide the pointwise confidence band as a reference.

\begin{figure}[t!]
	\centering
	\begin{subfigure}{0.31\textwidth}
		\centering
		\includegraphics[width = \textwidth]{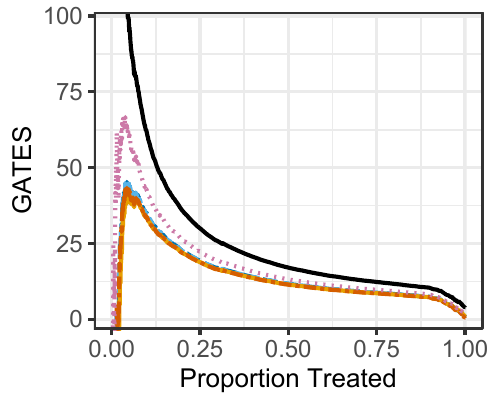}
		\caption{ \label{fig:conf_comp_full} Entire range $p
                  \in [0, 1]$}
	\end{subfigure}
	\begin{subfigure}{0.31\textwidth}
		\centering
		\includegraphics[width = \textwidth]{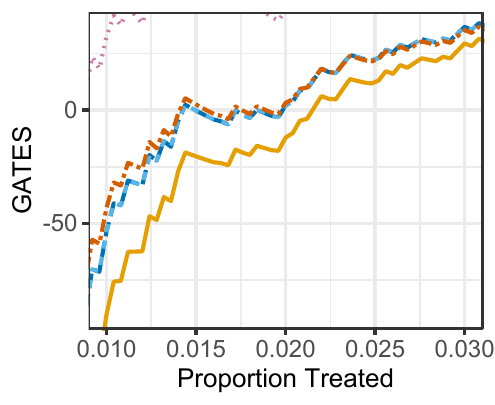}
		\caption{	\label{fig:conf_comp_0} Around $p=0.02$}
	\end{subfigure}
	\begin{subfigure}{0.31\textwidth}
		\centering
		\includegraphics[width = \textwidth]{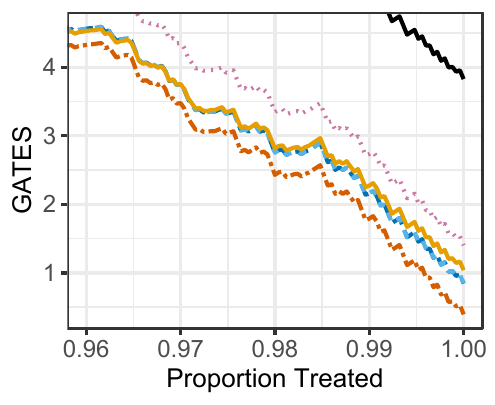}
		\caption{	\label{fig:conf_comp_1} Around $p=0.98$}
	\end{subfigure}
	\includegraphics[width=\textwidth]{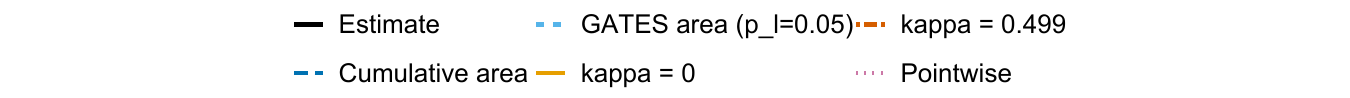}
	\caption{Comparison of uniform confidence bands
          using the sample of size $n=2500$. } \label{fig:conf_comp}
\end{figure}

Figure~\ref{fig:conf_comp_full} shows the resulting confidence bands
with the solid black line representing the point estimate
$\widehat{\Psi}_n$, whereas
Figures~\ref{fig:conf_comp_0}~and~\ref{fig:conf_comp_1} zoom into the
areas around $p=0.02$ and $p=0.98$, respectively. We find that in
this numerical example, the confidence band based on
Equation~\eqref{eq:k_family} with $\kappa=0$ (solid orange line) has the
greatest width when $p$ is small although it is slightly tighter than
the other uniform confidence bands when $p$ is large.
The confidence band based on Equation~\eqref{eq:k_family} with
$\kappa=0.499$ is approximately $4\%$ tighter than the minimum-area
bands at $p=0.01$, although it is approximately $15\%$ wider at $p=1$.
The two minimum-area confidence bands coincide in this example.

We also observe that, over $p\in[0.05,1]$, the different types of uniform
confidence bands are on average approximately $39\%$ to $46\%$ wider
than the pointwise confidence bands. This suggests that the uniform bands can
be practically useful while retaining a strong statistical
guarantee. The uniform bands, relative to the pointwise ones, are
generally wider when $p$ becomes closer to zero, and tighter when $p$
gets closer to one, where less uncertainty arises from estimating
the subgroup threshold.

\subsection{Extension to Other Evaluation Metrics}
\label{sec:sample_splitting_gen}

Our analysis so far evaluates selected subgroups using their average treatment effect. The same construction applies to a broader class of evaluation metrics that can be written as contrasts between treatment specific outcomes together with a function of observed covariates. For $t\in\{0,1\}$, let
\[
R_i(t)=g_t\{Y_i(t),\bX_i\},
\]
where $g_t$ is prespecified, and let $H(\bX_i)$ be a known function of the covariates. Define
\begin{align}
	\psi_i^g
	&=R_i(1)-R_i(0)+H(\bX_i),\nonumber\\
	\widehat\psi_i^g
	&=\frac{n}{n_1}T_iR_i(1)
	-\frac{n}{n_0}(1-T_i)R_i(0)
	+H(\bX_i),\nonumber\\
	\Psi^g(p)
	&=\E\{\psi_i^g\mid F(S_i)\geq1-p\},\qquad
	\widehat\Psi_n^g(p)
	=\frac{1}{np}\sum_{i=1}^{\lfloor np\rfloor}\widehat\psi_{[n,i]}^g.
	\label{eq:general_metric}
\end{align}
The next proposition shows that the uniform confidence bands developed above continue to apply to this class of metrics.

\begin{proposition}[Other Evaluation Metrics]
	\label{prop:general_metric}\spacingset{1}
	Suppose Assumptions~\ref{asm:randomsample}--\ref{asm:cdf} hold,
	\[
	\E\{R_i(1)^2+R_i(0)^2+H(\bX_i)^2\}<\infty,
	\]
	and
	\[
	\lambda_g(p)=\E\{\psi_i^g\mid F(S_i)=1-p\}
	\]
	is continuous on $[0,1]$. Let $c_g=\E(\psi_i^g)$ and define
	\begin{align*}
		r_g(p)
		&=\lim_{n\to\infty}\left[
		\frac{n}{n_1}\E\{R_i(1)^2h_p(S_i,Y_i(1))\}
		+\frac{n}{n_0}\E\{R_i(0)^2h_p(S_i,Y_i(0))\}\right]\\
		&\quad+2\E\{H(\bX_i)\psi_i^g h_p(S_i,Y_i(0))\}
		-\E\{H(\bX_i)^2h_p(S_i,Y_i(0))\}
		-\int_0^p\lambda_g(u)^2\,\d u,\\
		d_g(p)
		&=\lim_{n\to\infty}
		\frac{1}{\sqrt{n_1n_0}}\bigl[n_0\E\{R_i(1)h_p(S_i,Y_i(1))\}\\
		&\hspace{42mm}+n_1\E\{R_i(0)h_p(S_i,Y_i(0))\}\bigr],\\
		v_g(p)
		&=r_g(p)+\int_0^p\{\lambda_g(u)-c_g\}^2\,\d u.
	\end{align*}
	If $v_g(1)>0$, Theorem~\ref{thm:sorteduniform} and Corollary~\ref{cor:conf_band_k_family} apply to $\Psi^g$ after replacing $\widehat\psi_i$ by $\widehat\psi_i^g$ throughout the construction.
\end{proposition}

The proof is given in Appendix~\ref{app:general_metric}. Proposition~\ref{prop:general_metric} allows the subgroup to be evaluated using criteria that incorporate information observed for every unit in addition to the randomized outcome contrast. One useful example is treatment cost. If $c_0(\bX_i)$ denotes a known individual level treatment cost, setting
\[
R_i(t)=Y_i(t),
\qquad
H(\bX_i)=-c_0(\bX_i)
\]
gives
\[
\Psi^C(p)
=
\E\{Y_i(1)-Y_i(0)-c_0(\bX_i)
\mid F(S_i)\geq1-p\}.
\]
The corresponding uniform band therefore evaluates the net treatment benefit of each candidate subgroup. The same result can be used to evaluate characteristics of the selected population itself. For example, setting $R_i(t)=0$ and $H(\bX_i)=X_{ij}$ yields a uniform band for the average value of the baseline covariate $X_{ij}$ among the top $p$ proportion of units. Because this quantity is observed for every unit, $d_g(p)=0$.

The cost can also depend on the selected subgroup proportion. Suppose
\[
c_p(\bx)=\sum_{j=1}^m a_j(p)c_j(\bx),
\]
where $a_j$ and $c_j$ are known. Let $L$ denote a lower confidence band for $\Psi$, and let $L_j$ and $U_j$ denote jointly calibrated confidence bands for
\[
\E\{c_j(\bX_i)\mid F(S_i)\geq1-p\}.
\]
A uniform lower confidence band for the corresponding cost adjusted treatment effect is then
\begin{equation}
	L^C(p)
	=
	L(p)-\sum_{j=1}^m
	\{a_j(p)^+U_j(p)-a_j(p)^-L_j(p)\},
	\label{eq:proportion_dependent_cost}
\end{equation}
where $a^+=\max(a,0)$ and $a^-=\max(-a,0)$.

A different class of evaluation criteria compares the treatment effect within the selected subgroup with the treatment effect in the overall population. Motivated by the population average prescriptive effect, consider
\begin{equation}
	\Psi^M(p)=\Psi(p)-p\Psi(1).
	\label{eq:mean_adjusted_target}
\end{equation}
Its estimator is
\begin{equation}
	\widehat\Psi_n^M(p)
	=
	\widehat\Psi_n(p)
	-\frac{\lfloor np\rfloor}{n}\widehat\Psi_n(1),
	\label{eq:mean_decompose}
\end{equation}
which equals $\widehat\Psi_n(p)-p\widehat\Psi_n(1)$ at the observed rank fractions $p=k/n$. Since
\[
A^M(p)=p\Psi^M(p)=A(p)-p^2A(1),
\]
Proposition~\ref{prop:donsker} immediately gives
\begin{equation}
	G^M(p)=G(p)-p^2G(1).
	\label{eq:mean_adjusted_limit}
\end{equation}
The same transformation can be applied to the reference process,
\[
\widehat G_n^M(p)
=
\widehat G_n(p)-p^2\widehat G_n(1).
\]
Hence, the Gaussian reference construction developed above carries over directly to the mean adjusted criterion.

An important difference is that $G^M(1)=0$. We therefore use a boundary that also contracts at the upper endpoint. Let
\[
\widehat t_n(p)
=
\frac{\widehat v_n(p)}{\widehat v_n(1)}
\]
and define
\begin{equation}
	\widehat h^M_{\delta,n}(p)
	=
	\sqrt{\widehat v_n(1)}
	\left[
	\delta_0+
	\delta_1\sqrt{\widehat t_n(p)\{1-\widehat t_n(p)\}}
	\right],
	\qquad
	\delta=(\delta_0,\delta_1).
	\label{eq:mean_shape}
\end{equation}
For a fixed $\epsilon>0$, choose $(\widehat\delta_0^*,\widehat\delta_1^*)$ by solving
\begin{equation}
	\begin{aligned}
		(\widehat\delta_0^*,\widehat\delta_1^*)
		&=
		\argmin_{\delta\in\mathcal B_\epsilon}
		\int_0^1\widehat h^M_{\delta,n}(p)\,\d p\\
		\text{subject to}\qquad
		&\P^*\{\widehat G_n^M(p)
		\leq\widehat h^M_{\delta,n}(p),
		\ \forall p\in[0,1]\}
		\geq1-\alpha.
	\end{aligned}
	\label{eq:mean_CIoptim}
\end{equation}

\begin{proposition}[Uniform Band for the Mean Adjusted GATES]
	\label{prop:sortedunif_mean}\spacingset{1}
	Under the conditions of Theorem~\ref{thm:sorteduniform}, define
	\begin{equation}
		C_n^M(p,\alpha)
		=
		\frac{\sqrt{\widehat v_n(1)}}{p\sqrt n}
		\left[
		\widehat\delta_0^*
		+\widehat\delta_1^*
		\sqrt{\widehat t_n(p)\{1-\widehat t_n(p)\}}
		\right],
		\qquad 0<p<1,
		\label{eq:mean_uniform_band}
	\end{equation}
	and set $C_n^M(1,\alpha)=0$. Then
	\begin{equation}
		\lim_{n\to\infty}\P\left\{
		\Psi^M(p)\geq
		\widehat\Psi_n^M(p)-C_n^M(p,\alpha),
		\ \forall p\in(0,1]
		\right\}
		=1-\alpha.
		\label{eq:mean_uniform_coverage}
	\end{equation}
	Moreover, the fitted cumulative boundary is asymptotically minimum area within $\mathcal B_\epsilon$.
\end{proposition}

The proof is given in Appendix~\ref{app:sortedunif_mean}. Since both $\Psi^M(1)$ and $\widehat\Psi_n^M(1)$ equal zero, the band has zero width at $p=1$. Computationally, no new Gaussian simulation is required: the draws used for the original GATES band can be transformed according to Equation~\eqref{eq:mean_adjusted_limit} and recalibrated using Equation~\eqref{eq:mean_CIoptim}.

The same argument extends to other transformations that impose an endpoint restriction. Let $a:[0,1]\to[0,1]$ be continuous with $a(0)=0$ and $a(1)=1$, and define
\[
A^a(p)=A(p)-a(p)A(1).
\]
The corresponding limiting and reference processes are obtained by replacing $p^2$ in Equation~\eqref{eq:mean_adjusted_limit} with $a(p)$, after which the boundary is calibrated using the transformed reference process. In particular, $a(p)=p$ yields the cumulative prescriptive effect and, after dividing by $p$, the targeted operating characteristic $\Psi(p)-\Psi(1)$ studied by \citet{imai2021experimental,yadlowsky2021evaluating}.

\subsection{Characterizing the Selected Subgroup}
\label{subsec:interpretable}

In practice, researchers may wish to characterize the subgroup identified by the treatment prioritization score in addition to evaluating its treatment effect. This can be done in two ways within our framework.

One approach is to restrict the scoring rule to an interpretable function class, such as a decision tree or a linear model. Since our inference does not require the scoring rule to be consistent for an underlying population target, imposing such a restriction does not affect the validity of the confidence bands. Once the scoring rule has been estimated and fixed, the methodology above applies directly to the resulting GATES.

A second approach is to characterize the selected subgroup using its observed covariates \citep{cher:etal:19}. For covariate $X_{ij}$, define
\begin{equation}
	\Psi^{X_j}(p)
	=
	\E\{X_{ij}\mid F(S_i)\geq1-p\},
	\qquad
	\widehat\Psi_n^{X_j}(p)
	=
	\frac{1}{np}
	\sum_{i=1}^{\lfloor np\rfloor}X_{j,[n,i]}.
	\label{eq:covariate_gates}
\end{equation}
Here, $\Psi^{X_j}(p)$ gives the average value of $X_j$ among the top $p$ proportion of units under the treatment prioritization score. Proposition~\ref{prop:general_metric} applies by setting $R_i(t)=0$ and $H(\bX_i)=X_{ij}$, and therefore yields a uniform confidence band for $\Psi^{X_j}(p)$. Applying the same construction to $-X_{ij}$ gives the corresponding upper confidence band.

These covariate bands can be combined with the GATES band to characterize both the performance and composition of the selected subgroup. For example, if the GATES band is constructed at level $\alpha_0$ and the lower band for $X_j$ at level $\alpha_j$, then
\begin{align}
	\liminf_{n\to\infty}\P\Bigl\{
	&\Psi(p)\geq
	\widehat\Psi_n(p)-C_n(p,\alpha_0),\nonumber\\
	&\Psi^{X_j}(p)\geq
	\widehat\Psi_n^{X_j}(p)-C_n^{X_j}(p,\alpha_j),
	\quad \forall p\in(0,1]
	\Bigr\}
	\geq1-\alpha_0-\alpha_j.
	\label{eq:joint_covariate_band}
\end{align}
The same construction can be applied jointly to any fixed collection of covariates by allocating the desired error probability across the corresponding bands. Since all of these bands are uniform in $p$, the resulting guarantees remain valid when the subgroup proportion is selected using the evaluation sample.

If the objective is instead to compare the selected subgroup with the overall population, the endpoint adjusted construction in Section~\ref{sec:sample_splitting_gen} can be applied to
\[
\Psi^{X_j}(p)-\Psi^{X_j}(1).
\]
This gives a uniform characterization of how the covariate composition of the selected subgroup differs from that of the population as the subgroup proportion varies.

\subsection{Estimation Uncertainty of Scoring Rule}
\label{subsec:uncertainty}

The preceding results evaluate the scoring rule used to order the evaluation sample. This rule may itself be estimated from an independent training sample, and no consistency of the estimated rule is required for inference on its performance. The target is therefore the performance of the fitted rule that would actually be used to select a subgroup.

Researchers may instead wish to make inference on the performance of a fixed population scoring rule $f$. Let $\widehat f$ be estimated using data independent of the evaluation sample, and let $\widehat A_n^{\widehat f}$ and $\widehat A_n^f$ denote the cumulative estimators obtained by ordering the same evaluation sample using $\widehat f$ and $f$, respectively. The corresponding population target is
\begin{equation}
	A^f(p)
	=
	\E\left[
	\psi_i\mathbf{1}\{F_f(f(\bX_i))\geq1-p\}
	\right],
	\label{eq:population_rule_target}
\end{equation}
where $F_f$ denotes the distribution function of $f(\bX_i)$. Inference for $A^f$ therefore requires the ordering induced by $\widehat f$ to be asymptotically equivalent to that induced by $f$.

\renewcommand{\theassumption}{3A}
\begin{assumption}[Approximation of the Population Scoring Rule]
	\label{asm:cdf_uncertain}\spacingset{1}
	The population score $f(\bX)$ has a continuous distribution, and the fitted score $\widehat f(\bX)$ has a continuous distribution conditional on the independent training sample. There exists a deterministic sequence $r_n\downarrow0$ such that
	\begin{equation}
		\P\left(
		\sup_{\bx\in\cX}
		|\widehat f(\bx)-f(\bx)|
		\leq r_n
		\right)
		\longrightarrow1.
		\label{eq:score_uniform_convergence}
	\end{equation}
\end{assumption}

\renewcommand{\theassumption}{4A}
\begin{assumption}[Regularity near Score Thresholds]
	\label{asm:continuity_uncertain}\spacingset{1}
	The conditional mean treatment effect given the rank of $f(\bX_i)$ has a continuous version on $[0,1]$. Moreover,
	\begin{equation}
		\sqrt n\sup_{s\in\mathbb R}
		\E\left[
		\{|Y_i(1)|+|Y_i(0)|\}
		\mathbf{1}\{|f(\bX_i)-s|\leq4r_n\}
		\right]
		\longrightarrow0.
		\label{eq:weighted_margin}
	\end{equation}
\end{assumption}
\renewcommand{\theassumption}{\arabic{assumption}}

Assumption~\ref{asm:continuity_uncertain} controls the effect of observations whose relative ordering can change when $f$ is replaced by $\widehat f$. Only observations lying near a score threshold can switch sides of a selected subgroup, and Equation~\eqref{eq:weighted_margin} requires their contribution to the cumulative outcome to be negligible at the $n^{-1/2}$ scale.

The following proposition shows that, under these conditions, estimation of the scoring rule has no first-order effect on the uniform inference developed above.

\begin{proposition}[Inference for a Population Scoring Rule]
	\label{prop:donsker_uncertain}\spacingset{1}
	Suppose Assumptions~\ref{asm:randomsample}, \ref{asm:comrand},
	\ref{asm:moments}, \ref{asm:cdf_uncertain}, and
	\ref{asm:continuity_uncertain} hold. Then
	\begin{equation}
		\sqrt n
		\sup_{p\in[0,1]}
		\left|
		\widehat A_n^{\widehat f}(p)
		-
		\widehat A_n^f(p)
		\right|
		\xrightarrow{\P}0.
		\label{eq:score_equivalence}
	\end{equation}
	If $v_c^f(1)>0$ and $J_n\to\infty$ with $J_n^2/n\to0$, the estimators of the reference process computed using the ordering induced by $\widehat f$ converge uniformly to their counterparts under $f$. Consequently, Theorem~\ref{thm:sorteduniform}, Corollary~\ref{cor:conf_band_k_family}, and Proposition~\ref{prop:sortedunif_mean} continue to hold for the population-rule targets when implemented using $\widehat f$.
\end{proposition}

The proof is given in Appendix~\ref{app:donsker_uncertain}. The proposition formalizes the distinction between evaluating the fitted rule and evaluating an underlying population rule. The former requires no consistency of the scoring model, while the latter requires its estimation error to be small enough that perturbations of the induced ordering are asymptotically negligible.

\section{A Simulation Study}
\label{sec:synthetic}

We conduct a simulation study to examine the finite sample performance
of the proposed methodology. We use covariates from the 2016 Atlantic
Causal Inference Conference (ACIC) Competition.
We briefly describe its simulation setting here and refer interested
readers to \citet{dori:etal:19} for additional details.  The focus of
this competition was the inference of the average treatment effect in
observational studies. There are a total of 58 pre-treatment
covariates, taken from a real-world study with sample size $n=4802$.

In our simulation, we assume that the empirical distribution of these
covariates represent the population of interest and obtain each
simulation sample via bootstrap. One important change from the
original competition is that instead of utilizing a propensity model
to determine $T$, we assume that the treatment assignment is
completely randomized, i.e., $\Pr(T_i = 1)=1/2$, and the treatment and
control groups are of equal size, i.e., $n_1=n_0=n/2$.  

To generate the outcome variable, we evaluate the following
generalized additive model with polynomial basis functions directly,
using the covariate transformations in the ACIC generator.  The model
represents a setting, in which there exists a substantial amount of
treatment effect heterogeneity.  For the sake of completeness, the
formula for this outcome model is reproduced here:
{\footnotesize
  \begin{align*} \E(Y_i(t) \mid \bX_i)
    &= 1.60+ 0.53\times
      x_{29}-3.80\times x_{29}(x_{29}-0.98)(x_{29}+0.86) -0.32 \times
      \bone\{x_{17}>0\}\\& + 0.21 \times \bone\{x_{42}>0\}-0.63 \times
    x_{27}+4.68 \times \bone\{x_{27}<-0.61\}-0.39 \times (x_{27}+0.91)
    \bone\{x_{27}<-0.91\}\\&+ 0.75 \times \bone\{x_{30}\leq0\}-1.22
    \times \bone\{x_{54}\leq0\}+0.11 \times x_{37}
    \bone\{x_{4}\leq0\}-0.71 \times \bone\{x_{17}\leq0, t=0\}\\&-1.82
    \times \bone\{x_{42}\leq 0,t=1\}+0.28 \times \bone\{x_{30}\leq
    0,t=0\}\\&+\{0.58\times x_{29}-9.42 \times
    x_{29}(x_{29}-0.67)(x_{29}+0.34)\}\times\bone\{t=1\}\\&+(0.44
    \times x_{27}-4.87\times
    \bone\{x_{27}<-0.80\})\times\bone\{t=0\}-2.54 \times \bone\{t=0,
    x_{54}\leq 0\}.
  \end{align*}
} for $t=0,1$ where the error term is distributed according to the
standard normal distribution.

We estimate the CATE and use it as our treatment prioritization score
to identify exceptional responders.  For the estimation of the CATE,
we use Bayesian Additive Regression Trees (BART)
\citep[see][]{chipman2010bart,hill2011bayesian,hahn2020bayesian} and
Causal Forest \citep{wager2018estimation}, and LASSO
\citep{tibs:96}. The models are trained by fitting them to the
population data (i.e., the empirical distribution of the full
sample). For implementation, we use {\sf R
  4.4.3} with the following packages: {\sf bartMachine} (version
1.3.4.1) for BART, {\sf grf} (version 2.5.0) for Causal Forest, and {\sf
  glmnet} (version 4.1.10) for LASSO. We retain the fixed fitting
settings in the supplied simulation code, using 100 trees for BART
and 2,000 trees for Causal Forest. The LASSO penalty is $0.05$.
For each model and sample size, we conduct 1,000 simulation repetitions
and calibrate each uniform band with 10,000 Gaussian draws, using the
block rule and coefficient grid described in Section~\ref{subsec:numeric}.
Coverage is evaluated over the entire interval $p\in(0,1]$ against
the exactly integrated empirical-population GATES curve. We check
both ends of every sample rank interval and all population knots
within it. We use ordinary empirical calibration quantiles.

\begin{table}[t!]
  \centering \footnotesize
  \resizebox{\textwidth}{!}{\begin{tabular}{l|...|...|...}
    \hline 
    ML algorithm & \multicolumn{3}{c|}{Uniform} & \multicolumn{3}{c|}{Pointwise}& \multicolumn{3}{c}{Pointwise$\times 1.5$}\\
    & \multicolumn{1}{c}{$\bm{n=100}$} &
               \multicolumn{1}{c}{$\bm{n=500}$} & \multicolumn{1}{c|}{$\bm{n=2500}$} & \multicolumn{1}{c}{$\bm{n=100}$} &
               \multicolumn{1}{c}{$\bm{n=500}$} & \multicolumn{1}{c|}{$\bm{n=2500}$}&
               \multicolumn{1}{c}{$\bm{n=100}$} &
               \multicolumn{1}{c}{$\bm{n=500}$} & \multicolumn{1}{c}{$\bm{n=2500}$}\\
    \hline
    BART & 97.8\% & 97.4\% & 96.9\% & 94.8\% & 91.2\% & 86.6\% & 99.8\% & 99.8\% & 98.9\% \\
    Causal Forest & 98.1\% & 97.2\% & 96.8\% & 94.4\% & 91.2\% & 86.7\% & 99.7\% & 99.6\% & 98.8\% \\
    LASSO & 98.0\% & 97.6\% & 97.2\% & 94.8\% & 91.7\% & 88.3\% & 99.7\% & 99.9\% & 98.7\% \\\hline
  \end{tabular}}
      \spacingset{1}
\caption{Empirical coverage of 95\% lower confidence bands.
  The table compares the uniform band in
  Theorem~\ref{thm:sorteduniform} with the pointwise band and the
  pointwise band inflated by $50\%$. Coverage is the percentage of
  1,000 experiments in which the band contains the entire GATES curve,
  including between observed rank fractions.} \label{tb:simulation}

\end{table}

Table~\ref{tb:simulation} shows the proportion of simulation
experiments in which the confidence band contains the entire curve of
$\Psi^\ast(p)$. The uniform bands attain coverage between $96.8\%$
and $98.1\%$, exceeding the nominal level even at $n=100$.
The pointwise confidence bands have coverage between $86.6\%$ and
$94.8\%$ and generally become less reliable for simultaneous coverage
as the sample size grows.

Following our observations in Section~\ref{subsec:numeric}, we also
show the coverage of the pointwise confidence bands after inflating
their widths by $50\%$. This adjustment produces coverage between
$98.7\%$ and $99.9\%$ in this simulation, making it conservative
in this design. The directly calibrated uniform band uses the joint
reference distribution to determine the required width.

\begin{figure}[t!]
	\centering \spacingset{1}
	\includegraphics[width=\textwidth]{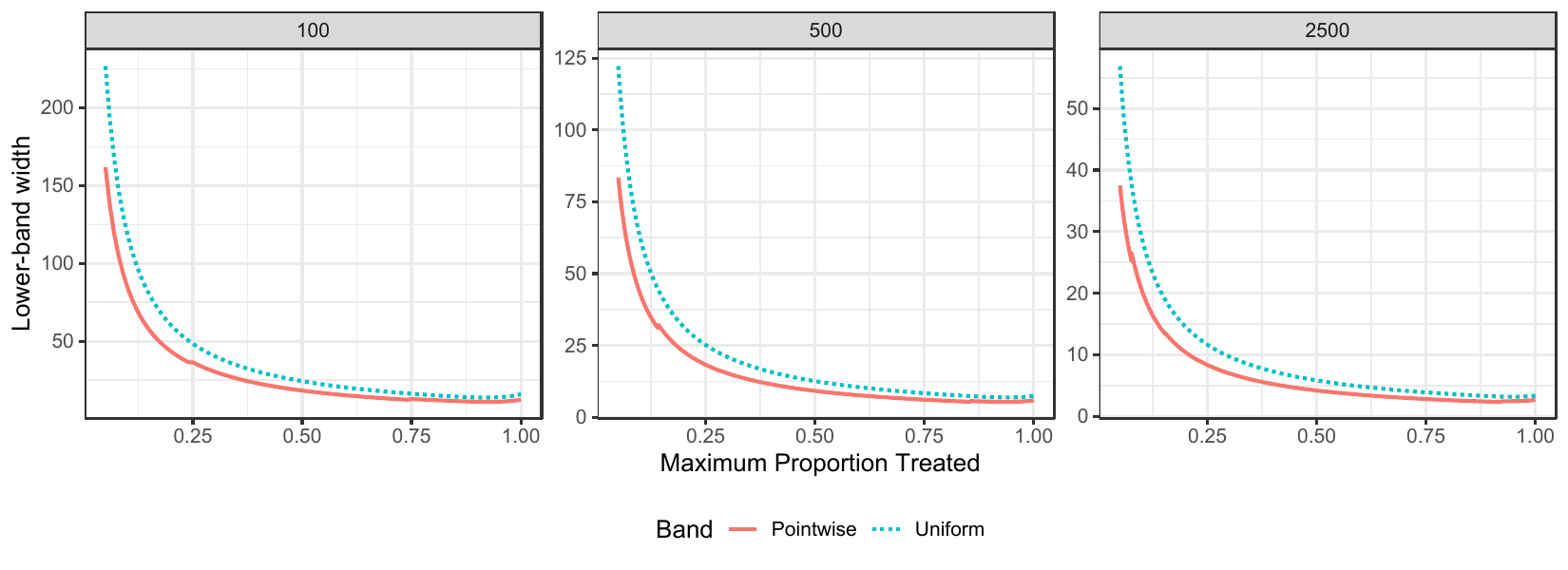}
	\caption{Comparison of the widths of the 95\% uniform and
          pointwise lower confidence bands, averaged over 1,000
          simulation repetitions. The treatment prioritization score
          is constructed using LASSO. The horizontal axis is the
          maximum proportion treated, displayed over $p\in[0.05,1]$.}

\label{fig:ci_length_lasso}
\end{figure}

Finally, Figure~\ref{fig:ci_length_lasso} compares the widths of
these uniform and pointwise confidence bands where the treatment
prioritization score is estimated using LASSO. Over $p\in[0.05,1]$,
the average uniform band is approximately $36\%$ to $41\%$ wider than
the pointwise band across the three sample sizes. The additional
width is greatest for small subgroups, where simultaneous inference
must account for uncertainty across the score ordering.

\section{Empirical Application}
\label{sec:realworld}

We apply the proposed methodology to the data from a clinical trial of
late-stage prostate cancer. 

\subsection{Setup}

\cite{byar1980choice} present the data from a randomized
placebo-controlled clinical trial to estimate the effect of
diethylstilbestrol on overall survival of patients. Three different
dosage levels (0.2 mg, 1.0mg, or 5.0mg) were used in the trial that
involved 502 patients with stage 3 or 4 prostate cancer. For each
patient, the trial recorded their baseline health characteristics,
common laboratory measurements, prior disease history, and detailed
information on the current prostate cancer. The primary outcome is the
number of months of total survival at the end of follow-up, which may
have occurred at either death or the program completion.

We will focus on the 5.0mg subgroup and consider the binary treatment
of giving either 5.0mg estrogen or placebo for these patients. In
total, 127 patients received the placebo while 125 patients received
the estrogen treatment. The average treatment effect was estimated to
be $-0.3$ months but this estimate was not statistically significant,
suggesting that on average the estrogen treatment do not improve total
survival.

By analyzing this data, we wish to identify a subgroup that may
significantly benefit from the treatment although the overall average
treatment effect estimate was not statistically significant. We
utilize a total of 13 pre-treatment covariates, including age, weight,
height, blood pressure, tumor size and stage, ECG diagnosis, and blood
measurements. The estrogen treatment is denoted as $T_i=1$ whereas the
placebo is represented as $T_i=0$. The outcome variable $Y_i$ is the
total survival.

We randomly selected approximately 40\% of the sample (i.e., 103
observations) as the training data and the remainder of the sample
(i.e., 149 observations) as the evaluation data. The split fixes the
training counts at 51 treated patients and 52 control patients and
does not use survival outcomes. We train four machine learning
models: the same three models as those used in
Section~\ref{sec:synthetic} and the Causal Tree model
\citep{athey2016recursive}, which is included to illustrate how the
proposed methodology can be used to select interpretable subgroups
of exceptional responders.

We retain the fitting settings in the supplied implementation, using
700 trees for Causal Forest and 130 trees for BART. The LASSO penalty
is one, and the minimum node size for Causal Tree is ten. We add an
independent Gaussian perturbation with standard deviation $10^{-5}$
to each score before evaluating the ordering. We calibrate the
design-corrected bands using 50,000 Gaussian draws, with
$J_n=\lfloor n^{1/3}\rfloor$ blocks and $\epsilon=10^{-3}$.

\subsection{Findings}

\begin{figure}[t!]
	\centering \spacingset{1}
	\includegraphics[width = \textwidth]{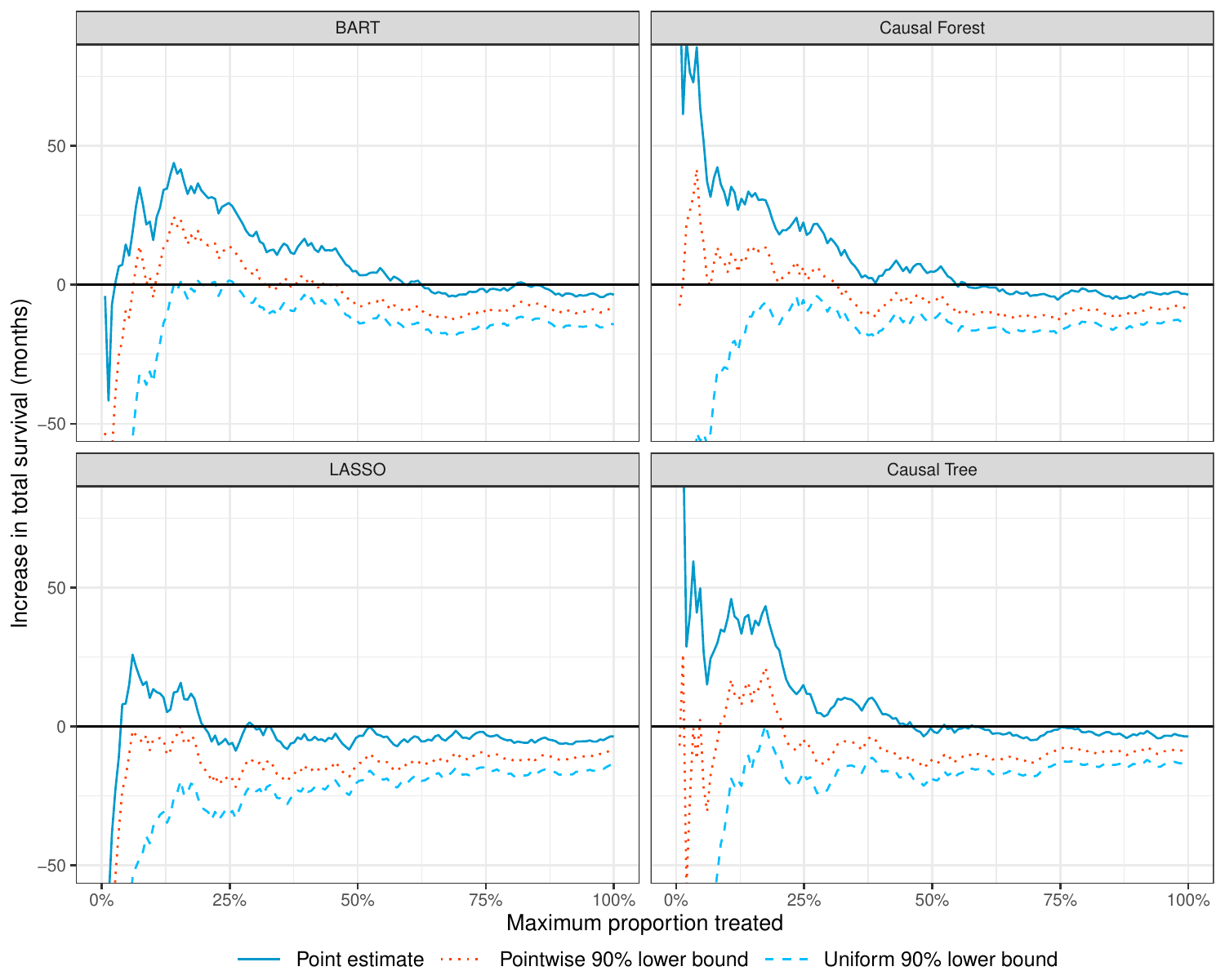}
	\caption{\label{fig:prostate} Evaluation of different machine
          learning algorithms on the randomized clinical trial for
          Diethylstilbestrol. The dark blue solid line represents the
          point estimates of the number of months of the increase in
          total survival when a proportion $p$ of the population is
          treated with Diethylstilbestrol. The dotted red line
          represents the pointwise $90\%$ confidence interval while
          the light blue dashed line represents the uniform $90\%$
          confidence interval. The black horizontal line represents
          zero treatment effect. These bands apply separately to each model
          and are also used in Table~\ref{tb:prostate_experiment}.}
\end{figure}

Figure~\ref{fig:prostate} shows the results. For each plot, we
compute the point estimates of the GATES (solid blue lines), the 90\%
minimum-area lower confidence bands (blue dashed lines) as shown in
Equation~\eqref{eq:CIoptim}, and 90\% pointwise confidence bands (red
dotted lines). We find that BART and Causal Tree produce positive
lower bounds for some subgroup sizes when each model is considered
separately. The lower bound reaches 1.58 months for BART at $p=24.8\%$
and 0.62 months for Causal Tree at $p=17.4\%$.
Causal Forest and LASSO do not produce positive uniform
lower bounds. We also find that the minimum-area uniform confidence
bands are closer to the pointwise bands for larger values of $p$ and
diverge more substantially for small values of $p$.

\begin{table}[t!]
	\centering
	\spacingset{1}
        \begin{tabular}{l|..r}
          \hline
          & \multicolumn{1}{c}{Estimated proportion of} &
                                                          \multicolumn{1}{c}{Estimated}
          & \multicolumn{1}{c}{90\% uniform}\\
          Estimator
          & \multicolumn{1}{c}{selected subgroup} &
                                                         \multicolumn{1}{c}{GATES}&
                                                                                    \multicolumn{1}{c}{confidence band}	\\ \hline
          Causal Forest & 27.5\% & 22.0 & ($-$4.02, $\infty$) \\
          BART & 24.8\% & 29.5 & (1.58, $\infty$) \\
          LASSO & 100.0\% & $-$3.52 & ($-$13.72, $\infty$) \\
          Causal Tree & 17.4\% & 43.4 & (0.62, $\infty$) \\\hline

        \end{tabular}
	\caption{Estimated proportion of the selected subgroup,
          their GATES estimates, and 90\% uniform confidence
          bands for the subgroup maximizing the lower bound. Each row represents the results based on different
          machine learning algorithms used to estimate the treatment
          prioritization score. The confidence bands use significance level $0.10$ separately for each model, as in Figure~\ref{fig:prostate}. The outcome is total survival in
          months.} \label{tb:prostate_experiment}
\end{table}

Table~\ref{tb:prostate_experiment} presents the estimated subgroup
proportions, their GATES estimates, and the
corresponding 90\% uniform confidence bands, for each ML algorithm
used to construct the treatment prioritization score. The results
show that BART and Causal Tree identify subgroups with positive
uniform lower confidence bounds. BART selects $24.8\%$ of the population
with an estimated increase of 29.5 months in total survival and a lower
bound of 1.58 months. Causal Tree selects $17.4\%$ of the population
with an estimated increase of 43.4 months and a lower bound of 0.62
months. For each of these models, the uniform band supports a positive
average treatment effect after selecting the subgroup size.

In comparison, the previous optimization approach
\citep[e.g.][]{bertsimas2019identifying} found a subgroup
that could prolong total survival by 18 months on average but without
a statistical guarantee. Although the exact subgroup may vary across
training and evaluation splits, this example illustrates how the
uniform confidence band distinguishes an estimated benefit from one
supported after selection.

\begin{figure}[t!]
	\centering \spacingset{1}
	\includegraphics[width = 0.7\textwidth]{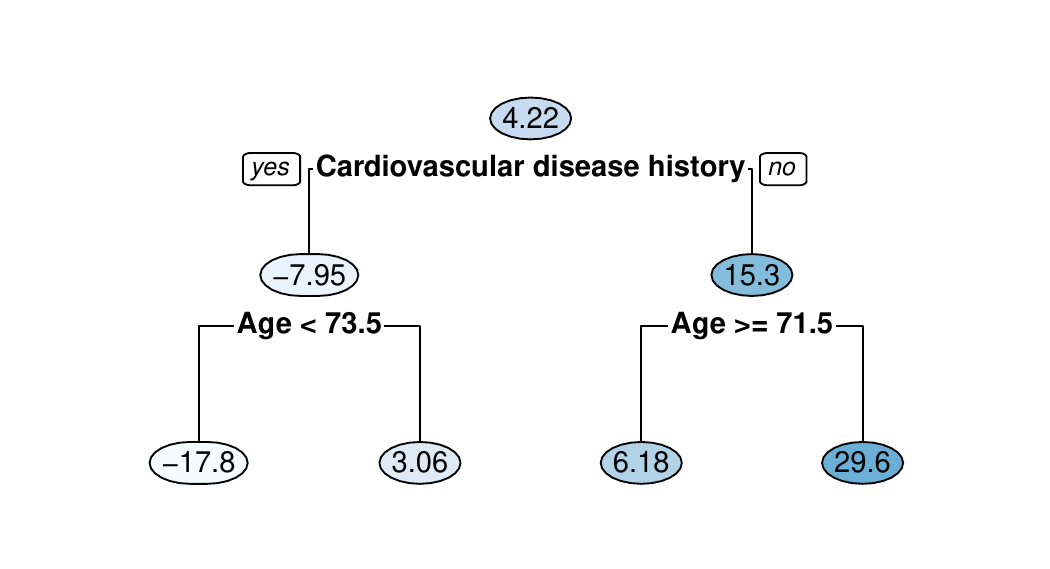}
	\caption{\label{fig:causaltree} The scoring function $\hat{S}$ according to the causal tree model. If the condition at the split is satisfied, the path continues left, and vice versa. The number represents the average treatment effect at the node. }
\end{figure}

A main advantage of Causal Tree is that it produces an interpretable
decision rule. As shown in Figure~\ref{fig:causaltree}, its scoring
function is a simple function of covariates. Its selected subgroup
contains $17.4\%$ of the evaluation sample and consists of patients
in the terminal node with the highest estimated treatment effect.
This yields a transparent description of the candidate subgroup:
patients younger than 71.5 years with no history of cardiovascular disease.

\begin{table}[t!]
  \centering
  \spacingset{1}
  \begin{tabular}{l|.r}
    \hline
Covariates
      &\multicolumn{1}{c}{Average Difference}&
                                         \multicolumn{1}{c}{ One-sided 90\% CI}	\\ \hline
    Tumor Size (cm$^2$) & 6.8 & $(1.9, \infty)$ \\
    Stage/Grade Index & 1.8 & $(1.1, \infty)$ \\
    Weight & -4.6 & $(-\infty, -0.6)$ \\ \hline

  \end{tabular}
  \caption{The difference in average values of covariates between the
    BART-selected subgroup and the overall population.  We report
    their (one-sided) 90\% uniform confidence bands. The confidence
    bands use significance level $0.10$ separately for each covariate. } \label{tb:features}
\end{table}

To interpret the results of the other models, we can follow the
strategy described in Section~\ref{subsec:interpretable} and
characterize the average values of various features within the
selected subgroup.  Table~\ref{tb:features} shows the average
difference between the selected subgroup and the population mean,
estimated using the evaluation sample, for three covariates with 90\% uniform
confidence bands.  We use the BART results for illustration.  We find that the average
tumor size of the selected group is 6.8 cm$^2$ larger than the
population average, and this is statistically significant at the
$\alpha=10\%$ level.  Similarly, we find that the stage/grade index, a
combination index of the tumor stage and histologic grade, is larger
by 1.8 points, with a positive lower confidence bound. The weight of
the subgroup is 4.6 units lower than the population average, with
an upper confidence bound of $-0.6$. This procedure
allows an interpretable characterization of the subgroup that could
lead to better understanding of the mechanisms behind exceptional
responders.

\section{Concluding Remarks}
\label{sec:conclude}

As the use of machine learning tools become widespread in high-stake
decision making, it is essential to deploy them effectively and
safely.  Doing so requires the accurate quantification of statistical
uncertainty and the robust evaluation of their empirical performance.
The proposed methodology provides a statistical performance guarantee
when selecting an optimal subset of a target population who are
predicted to benefit most from a treatment of interest.  We believe
that our methodology is particularly useful when the evaluation sample
is of limited size, machine learning models provide biased or noisy
estimates of treatment efficacy, and/or the proportion of exceptional
responders is relatively small.  In future research, we plan to extend
the proposed methodology to other settings, including non-binary and
dynamic treatments.  Our primary framework evaluates the fitted scoring rule that is
actually available for deployment.  Extending inference for a
hypothetical population rule beyond the sufficient conditions in
Section~\ref{subsec:uncertainty} is another direction for future
research.

\newpage
\pdfbookmark[1]{References}{References}
\bibliography{sample_resubmission,my,imai}

@Article{dori:etal:19,
  author = 	 {Dorie, Vincent and Hill, Jennifer and Shalit, Uri
                  and Scott, Marc and Cervone, Dan},
  title = 	 {Automated versus Do-It-Yourself Methods for Causal Inference: Lessons Learned from a Data Analysis Competition},
  journal = 	 {Statistical Science},
  year = 	 {2019},
  OPTkey = 	 {},
  volume = 	 {34},
  number = 	 {1},
  pages = 	 {43--68},
  month = 	 {February},
  OPTnote = 	 {},
  OPTannote = 	 {}
}

@article{neym:23,
	Author = {Neyman, J.},
	Journal = {Statistical Science},
	Pages = {465--480},
	Title = {On the application of probability theory to agricultural experiments: Essay on principles, Section 9. (Translated in 1990)},
	Volume = 5,
	Year = 1923}

@article{rubi:90,
	Author = {Rubin, Donald B.},
	Journal = {Statistical Science},
	Pages = {472--480},
	Title = {Comments on ``{O}n the Application of Probability Theory to Agricultural Experiments. {E}ssay on Principles. {S}ection 9'' by {J.} {S}plawa-{N}eyman translated from the {P}olish and edited by {D.} {M.} {D}abrowska and {T.} {P.} {S}peed},
	Volume = 5,
	Year = 1990}

@article{tibs:96,
	Author = {Tibshirani, Robert},
	Journal = {Journal of the Royal Statistical Society, Series {B} (Statistical Methodology)},
	Number = 1,
	Pages = {267--288},
	Title = {Regression Shrinkage and Selection via {LASSO}},
	Volume = 58,
	Year = 1996}

@Article{imai:li:22,
  author = 	 {Imai, Kosuke and Li, Michael Lingzhi},
  title = 	 {Statistical Inference for Heterogeneous Treatment Effects Discovered by Generic Machine Learning in Randomized Experiments},
  journal = 	 {Journal of Business \& Economic Statistics},
  year = 	 {2025},
  volume = 	 {43},
  number = 	 {1},
  pages = 	 {256--268},
  doi = 	 {10.1080/07350015.2024.2358909}
}

@article{10.1093/jnci/djaa061,
    author = {Conley, Barbara A and Staudt, Lou and Takebe, Naoko and Wheeler, David A and Wang, Linghua and Cardenas, Maria F and Korchina, Viktoriya and Zenklusen, Jean Claude and McShane, Lisa M and Tricoli, James V and Williams, Paul M and Lubensky, Irina and O’Sullivan-Coyne, Geraldine and Kohn, Elise and Little, Richard F and White, Jeffrey and Malik, Shakun and Harris, Lyndsay N and Mann, Bhupinder and Weil, Carol and Tarnuzzer, Roy and Karlovich, Chris and Rodgers, Brian and Shankar, Lalitha and Jacobs, Paula M and Nolan, Tracy and Berryman, Sean M and Gastier-Foster, Julie and Bowen, Jay and Leraas, Kristen and Shen, Hui and Laird, Peter W and Esteller, Manel and Miller, Vincent and Johnson, Adrienne and Edmondson, Elijah F and Giordano, Thomas J and Kim, Benjamin and Ivy, S Percy},
    title = "{The Exceptional Responders Initiative: Feasibility of a National Cancer Institute Pilot Study}",
    journal = {JNCI: Journal of the National Cancer Institute},
    volume = {113},
    number = {1},
    pages = {27-37},
    year = {2021},
    month = {01},
    issn = {0027-8874},
    doi = {10.1093/jnci/djaa061},
    url = {https://doi.org/10.1093/jnci/djaa061},
    eprint = {https://academic.oup.com/jnci/article-pdf/113/1/27/35388347/djaa061.pdf},
}

@article{cher:etal:19,
  author = 	 {Chernozhukov, Victor and Demirer, Mert and Duflo,
                  Esther and {Fernandez-Val}, Ivan},
  title = 	 {Fisher--Schultz Lecture: Generic Machine Learning Inference on Heterogeneous Treatment Effects in Randomized Experiments, With an Application to Immunization in India},
  journal = {Econometrica},
  year = {2025},
  volume = {93},
  number = {4},
  pages = {1121--1164},
  url = {https://economics.mit.edu/sites/default/files/2025-07/2025.07%20Fisher-Schultz%20Lecture.pdf}
}

@article{exceptional2015,
	author = {Takebe, Naoko and McShane, Lisa and Conley, Barbara},
	date = {2015/03/01},
	doi = {10.1038/nrclinonc.2015.19},
	id = {Takebe2015},
	isbn = {1759-4782},
	journal = {Nature Reviews Clinical Oncology},
	number = {3},
	pages = {132--134},
	title = {Exceptional responders---discovering predictive biomarkers},
	url = {https://doi.org/10.1038/nrclinonc.2015.19},
	volume = {12},
	year = {2015}}

@article{bhattacharya1974convergence,
  title={Convergence of Sample Paths of Normalized Sums of Induced Order Statistics},
  author={Bhattacharya, PK},
  journal={The Annals of Statistics},
  volume={2},
  number={5},
  pages={1034--1039},
  year={1974},
  doi={10.1214/aos/1176342823}
}

@article{davydov2000functional,
  title={Functional limit theorems for induced order statistics},
  author={Davydov, Youri and Egorov, V},
  journal={Mathematical Methods of Statistics},
  volume={9},
  number={3},
  pages={297--313},
  year={2000},
  publisher={New York: Allerton Press, c1992-}
}

@article{david1974asymptotic,
  title={The asymptotic theory of concomitants of order statistics},
  author={David, HA and Galambos, J},
  journal={Journal of Applied Probability},
  volume={11},
  number={4},
  pages={762--770},
  year={1974},
  publisher={Cambridge University Press}
}

@article{yang1977general,
  title={General Distribution Theory of the Concomitants of Order Statistics},
  author={Yang, SS},
  journal={The Annals of Statistics},
  volume={5},
  number={5},
  pages={996--1002},
  year={1977},
  publisher={Institute of Mathematical Statistics}
}

@article{sen1976note,
  title={A Note on Invariance Principles for Induced Order Statistics},
  author={Sen, Pranab Kumar},
  journal={The Annals of Probability},
  volume={4},
  number={3},
  pages={474--479},
  year={1976},
  publisher={Institute of Mathematical Statistics}
}

@article{imai2021experimental,
  author = 	 {Imai, Kosuke and Li, Michael Lingzhi},
  title = 	 {Experimental Evaluation of Individualized Treatment Rules},
  journal = 	 {Journal of the American Statistical Association},
  year = 	 {2023},
  OPTkey = 	 {},
  volume = 	 {118},
  number = 	 {541},
  pages = 	 {242--256},
  OPTmonth = 	 {},
  OPTnote = 	 {},
  OPTannote = 	 {}
}

@article{yadlowsky2021evaluating,
  title={Evaluating treatment prioritization rules via rank-weighted average treatment effects},
  author={Yadlowsky, Steve and Fleming, Scott and Shah, Nigam and Brunskill, Emma and Wager, Stefan},
  journal={Journal of the American Statistical Association},
  year={2025},
  volume={120},
  number={549},
  pages={38--51},
  doi={10.1080/01621459.2024.2393466}
}

@article{bertsimas2019identifying,
  title={Identifying exceptional responders in randomized trials: An optimization approach},
  author={Bertsimas, Dimitris and Korolko, Nikita and Weinstein, Alexander M},
  journal={Informs Journal on Optimization},
  volume={1},
  number={3},
  pages={187--199},
  year={2019},
  publisher={INFORMS}
}

@article{kehl2006responder,
  title={Responder identification in clinical trials with censored data},
  author={Kehl, Victoria and Ulm, Kurt},
  journal={Computational Statistics \& Data Analysis},
  volume={50},
  number={5},
  pages={1338--1355},
  year={2006},
  publisher={Elsevier}
}

@article{su2009subgroup,
  title={Subgroup analysis via recursive partitioning.},
  author={Su, Xiaogang and Tsai, Chih-Ling and Wang, Hansheng and Nickerson, David M and Li, Bogong},
  journal={Journal of Machine Learning Research},
  volume={10},
  number={5},
  pages={141--158},
  year={2009}
}

@article{hardin2013understanding,
  title={Understanding heterogeneity in response to antidiabetes treatment: a post hoc analysis using SIDES, a subgroup identification algorithm},
  author={Hardin, Dana S and Rohwer, Rachelle D and Curtis, Bradley H and Zagar, Anthony and Chen, Lei and Boye, Kristina S and Jiang, Honghua H and Lipkovich, Ilya A},
  journal={Journal of Diabetes Science and Technology},
  volume={7},
  number={2},
  pages={420--430},
  year={2013},
  publisher={SAGE Publications}
}

@article{foster2011subgroup,
  title={Subgroup identification from randomized clinical trial data},
  author={Foster, Jared C and Taylor, Jeremy MG and Ruberg, Stephen J},
  journal={Statistics in medicine},
  volume={30},
  number={24},
  pages={2867--2880},
  year={2011},
  publisher={Wiley Online Library}
}

@article{cho2021quantile,
  title={Quantile-based subgroup identification for randomized clinical trials},
  author={Cho, Youngjoo and Ghosh, Debashis},
  journal={Statistics in Biosciences},
  volume={13},
  pages={90--128},
  year={2021},
  publisher={Springer}
}

@article{athey2016recursive,
  title={Recursive partitioning for heterogeneous causal effects},
  author={Athey, Susan and Imbens, Guido},
  journal={Proceedings of the National Academy of Sciences},
  volume={113},
  number={27},
  pages={7353--7360},
  year={2016},
  publisher={National Acad Sciences}
}

@article{kendall2007confidence,
  title={Confidence bands for Brownian motion and applications to Monte Carlo simulation},
  author={Kendall, Wilfrid S and Marin, J-M and Robert, Christian P},
  journal={Statistics and Computing},
  volume={17},
  pages={1--10},
  year={2007},
  publisher={Springer}
}

@article{chipman2010bart,
  title={{BART}: Bayesian additive regression trees},
  author={Chipman, Hugh A and George, Edward I and McCulloch, Robert E},
  journal={The Annals of Applied Statistics},
  volume={4},
  number={1},
  pages={266--298},
  year={2010},
  publisher={Institute of Mathematical Statistics}
}

@article{hill2011bayesian,
  title={Bayesian nonparametric modeling for causal inference},
  author={Hill, Jennifer L},
  journal={Journal of Computational and Graphical Statistics},
  volume={20},
  number={1},
  pages={217--240},
  year={2011},
  publisher={Taylor \& Francis}
}

@article{hahn2020bayesian,
  title={Bayesian regression tree models for causal inference: regularization, confounding, and heterogeneous effects},
  author={Hahn, P Richard and Murray, Jared S and Carvalho, Carlos M},
  journal={Bayesian Analysis},
  year={2020},
  volume={15},
  number={3},
  pages={965--1056},
  publisher={International Society for Bayesian Analysis}
}

@article{radcliffe2007using,
  title={Using control groups to target on predicted lift: Building and assessing uplift models},
  author={Radcliffe, Nicholas J},
  journal={Direct Marketing Analytics Journal},
  volume={1},
  number={3},
  pages={14--21},
  year={2007}
}

@article{imai:ratk:13,
	Author = {Imai, Kosuke and Ratkovic, Marc},
	Journal = {Annals of Applied Statistics},
	Title = {Estimating Treatment Effect Heterogeneity in Randomized Program Evaluation},
	Year = {2013},
        Volume = {7},
        Number = {1},
        Month = {March},
        Pages = {443--470}}

@article{byar1980choice,
 	title={The choice of treatment for cancer patients based on covariate information.},
 	author={Byar, David P and Green, Sylvan B},
 	journal={Bulletin du cancer},
 	volume={67},
 	number={4},
 	pages={477--490},
 	year={1980}
 }

@article{ballarini2018subgroup,
 	title={Subgroup identification in clinical trials via the predicted individual treatment effect},
 	author={Ballarini, Nicol{\'a}s M and Rosenkranz, Gerd K and Jaki, Thomas and K{\"o}nig, Franz and Posch, Martin},
 	journal={PloS one},
 	volume={13},
 	number={10},
 	pages={e0205971},
 	year={2018},
 	publisher={Public Library of Science San Francisco, CA USA}
 }

@article{rosenkranz2016exploratory,
 	title={Exploratory subgroup analysis in clinical trials by model selection},
 	author={Rosenkranz, Gerd K},
 	journal={Biometrical Journal},
 	volume={58},
 	number={5},
 	pages={1217--1228},
 	year={2016},
 	publisher={Wiley Online Library}
 }

@article{bonetti2000graphical,
 	title={A graphical method to assess treatment--covariate interactions using the Cox model on subsets of the data},
 	author={Bonetti, Marco and Gelber, Richard D},
 	journal={Statistics in medicine},
 	volume={19},
 	number={19},
 	pages={2595--2609},
 	year={2000},
 	publisher={Wiley Online Library}
 }

@article{wager2018estimation,
  title={Estimation and inference of heterogeneous treatment effects using random forests},
  author={Wager, Stefan and Athey, Susan},
  journal={Journal of the American Statistical Association},
  volume={113},
  number={523},
  pages={1228--1242},
  year={2018},
  publisher={Taylor \& Francis}
}

@article{ossiander1987central,
  title={A central limit theorem under metric entropy with L2 bracketing},
  author={Ossiander, Mina},
  journal={The Annals of Probability},
  pages={897--919},
  year={1987},
  publisher={JSTOR}
}

@article{sussman2013using,
  title={Using benefit-based tailored treatment to improve the use of antihypertensive medications},
  author={Sussman, Jeremy and Vijan, Sandeep and Hayward, Rod},
  journal={Circulation},
  volume={128},
  number={21},
  pages={2309--2317},
  year={2013},
  publisher={Lippincott Williams \& Wilkins Hagerstown, MD}
}

@article{razonable2022clinical,
  title={Clinical prioritization of antispike monoclonal antibody treatment of mild to moderate COVID-19},
  author={Razonable, Raymund R and Ganesh, Ravindra and Bierle, Dennis M},
  journal={Mayo Clinic Proceedings},
  volume={97},
  number={1},
  pages={26--30},
  year={2022},
  publisher={Elsevier},
  doi={10.1016/j.mayocp.2021.11.017}
}

@article{miller1982maximally,
  title={Maximally selected chi square statistics},
  author={Miller, Rupert and Siegmund, David},
  journal={Biometrics},
  pages={1011--1016},
  year={1982},
  publisher={JSTOR}
}

@article{hothorn2002maximally,
  title={Maximally selected rank statistics in R},
  author={Hothorn, Torsten and Lausen, Berthold},
  journal={R News},
  volume={2},
  number={1},
  pages={3--5},
  year={2002}
}

@article{luedtke2016optimal,
  title={Optimal individualized treatments in resource-limited settings},
  author={Luedtke, Alexander R and van der Laan, Mark J},
  journal={The international journal of biostatistics},
  volume={12},
  number={1},
  pages={283--303},
  year={2016},
  publisher={De Gruyter}
}

@article{bonvini2023minimax,
  title={Minimax optimal subgroup identification},
  author={Bonvini, Matteo and Kennedy, Edward H and Keele, Luke J},
  journal={arXiv preprint arXiv:2306.17464},
  year={2023}
}

@article{imai2023experimental,
  title={Experimental evaluation of algorithm-assisted human decision-making: Application to pretrial public safety assessment},
  author={Imai, Kosuke and Jiang, Zhichao and Greiner, D James and Halen, Ryan and Shin, Sooahn},
  journal={Journal of the Royal Statistical Society Series A: Statistics in Society},
  volume={186},
  number={2},
  pages={167--189},
  year={2023},
  publisher={Oxford University Press US}
}

@article{lewis1999statistical,
  title={Statistical principles for clinical trials (ICH E9): an introductory note on an international guideline},
  author={Lewis, John A},
  journal={Statistics in medicine},
  volume={18},
  number={15},
  pages={1903--1942},
  year={1999},
  publisher={John Wiley \& Sons, Ltd. Chichester, UK}
}

@article{fogarty2026comparison,
  author = {Fogarty, Colin B.},
  title = {Gaussian comparison above the median},
  journal = {Bernoulli},
  year = {2026},
  note = {Forthcoming. arXiv:2607.06874v2, September 2, 2026},
  url = {https://arxiv.org/abs/2607.06874}
}

@article{massart1990comparison,
  author = {Massart, Pascal},
  title = {The tight constant in the {Dvoretzky--Kiefer--Wolfowitz} inequality},
  journal = {The Annals of Probability},
  year = {1990},
  volume = {18},
  number = {3},
  pages = {1269--1283},
  doi = {10.1214/aop/1176990746}
}

@article{prekopa1973joint,
  author = {Pr{\'e}kopa, Andr{\'a}s},
  title = {On logarithmic concave measures and functions},
  journal = {Acta Scientiarum Mathematicarum (Szeged)},
  year = {1973},
  volume = {34},
  pages = {335--343}
}

@book{bogachev1998joint,
  author = {Bogachev, Vladimir I.},
  title = {Gaussian Measures},
  series = {Mathematical Surveys and Monographs},
  volume = {62},
  publisher = {American Mathematical Society},
  address = {Providence, RI},
  year = {1998},
  doi = {10.1090/surv/062}
}

@article{cck2013design,
  author = {Chernozhukov, Victor and Chetverikov, Denis and Kato, Kengo},
  title = {Gaussian approximations and multiplier bootstrap for maxima of sums of high-dimensional random vectors},
  journal = {The Annals of Statistics},
  year = {2013},
  volume = {41},
  number = {6},
  pages = {2786--2819},
  doi = {10.1214/13-AOS1161}
}

@article{mccann1997direct,
  author = {McCann, Robert J.},
  title = {A Convexity Principle for Interacting Gases},
  journal = {Advances in Mathematics},
  year = {1997},
  volume = {128},
  number = {1},
  pages = {153--179},
  doi = {10.1006/aima.1997.1634}
}

\clearpage
\appendix
\spacingset{1}

\setcounter{table}{0}
\renewcommand{\thetable}{S\arabic{table}}
\setcounter{figure}{0}
\renewcommand{\thefigure}{S\arabic{figure}}
\setcounter{equation}{0}
\renewcommand{\theequation}{S\arabic{equation}}
\setcounter{theorem}{0}
\renewcommand {\thetheorem} {S\arabic{theorem}}
\setcounter{section}{0}
\renewcommand {\thesection} {S\arabic{section}}
\setcounter{lemma}{0}
\renewcommand {\thelemma} {S\arabic{lemma}}
\setcounter{proposition}{0}
\renewcommand {\theproposition} {S\arabic{proposition}}
\renewcommand{\theHequation}{S.\arabic{equation}}
\renewcommand{\theHtheorem}{S.\arabic{theorem}}
\renewcommand{\theHlemma}{S.\arabic{lemma}}
\renewcommand{\theHproposition}{S.\arabic{proposition}}
\renewcommand{\theHsection}{S.\arabic{section}}
\renewcommand{\theHfigure}{S.\arabic{figure}}
\renewcommand{\theHtable}{S.\arabic{table}}

\begin{center}
	\LARGE {\bf Supplementary Appendix}
\end{center}

\section{Proof of Proposition~\ref{prop:donsker}}
\label{app:donsker}

We begin by establishing the empirical process results used to prove Proposition~\ref{prop:donsker}. Throughout the appendices, we condition on the auxiliary data when the scoring rule is estimated from an independent sample. Let $P_t$ denote the law of $(S_i,Y_i(t))$ for $t\in\{0,1\}$. For a measurable function $h$ of the score and outcome, define
\[
 P_t h=\E\{h(S_i,Y_i(t))\},\qquad
 \mathbb P_{tn}h=\frac1{n_t}\sum_{i:T_i=t}h(S_i,Y_i).
\]
Thus, the functions indexing the arm processes have the same domain as in Equation~\eqref{eq:arm_gaussian_process}. We write $(S,Y(1),Y(0))$ for a generic sampled unit and set $q_n=1-\pi_n$ and $q=1-\pi$. To express the score ordering on $[0,1]$, define
\[
 U_i=1-F(S_i),\qquad U=1-F(S).
\]
These ranks are uniformly distributed by Assumption~\ref{asm:cdf}, and smaller ranks correspond to higher treatment priority. In particular,
\[
 h_p(S_i,Y_i(t))=\mathbf 1\{U_i\leq p\},\qquad
 \E(\psi_i\mid U_i=u)=\lambda(u).
\]
We continue to use $h_p(s,y)$ when the second argument is another observed mark, since the indicator does not depend on that argument.

\begin{lemma}[Treatment-arm Samples]\label{lem:arm_samples}
 Under Assumptions~\ref{asm:randomsample} and~\ref{asm:comrand}, the treatment and control samples, listed in their original order within each arm, have joint law $P_1^{n_1}\otimes P_0^{n_0}$.
\end{lemma}
\begin{proof}
We condition first on an assignment vector with exactly $n_1$ treated units. Under Assumption~\ref{asm:randomsample}, the observations assigned to the two arms are disjoint collections of independent population draws. Their conditional joint law is therefore $P_1^{n_1}\otimes P_0^{n_0}$. This law is the same for every admissible assignment vector. Averaging over the complete randomization distribution proves the result.
\end{proof}

\begin{lemma}[Marked Threshold Processes]\label{lem:marked_thresholds}
 Suppose $F(S)$ is uniformly distributed on $[0,1]$ and $\E R^2<\infty$. The classes
 \[
 \{(s,r)\mapsto h_p(s,r):0\leq p\leq1\},\qquad
 \{(s,r)\mapsto r h_p(s,r):0\leq p\leq1\}
 \]
 are jointly Donsker, with limiting processes continuous in $p$. The class $\{(s,r)\mapsto r^2h_p(s,r):0\leq p\leq1\}$ satisfies a uniform law of large numbers. These statements hold jointly for any fixed finite collection of square-integrable marks. Threshold classes with integrable marks satisfy the uniform law of large numbers.
\end{lemma}
\begin{proof}
We verify the bracketing entropy condition directly. Define the finite nonatomic measure
\[
 \nu(H)=\E\{(1+|R|)^2\mathbf 1(U\in H)\},\qquad H\subseteq[0,1].
\]
For $\epsilon>0$, partition $[0,1]$ into at most $C\epsilon^{-2}$ intervals of $\nu$-mass at most $\epsilon^2$. If $p\in(a,b]$ lies in one of these intervals, then
\[
 R h_a(S,R)-|R|\mathbf 1(a<U\leq b)
 \leq R h_p(S,R)
 \leq R h_a(S,R)+|R|\mathbf 1(a<U\leq b).
\]
The width of this bracket in $L_2$ is at most $2\epsilon$. The same partition brackets the unmarked indicators. Hence the bracketing entropy is bounded by $\log(C\epsilon^{-2})$, and
\[
 \int_0^1\sqrt{\log(C\epsilon^{-2})}\,\d\epsilon<\infty.
\]
The bracketing central limit theorem of \citet{ossiander1987central} therefore gives joint convergence of the marked and unmarked threshold processes.

To establish continuity of the limiting paths, note that
\[
 \E\left[R^2\{h_p(S,R)-h_{p'}(S,R)\}^2\right]
 =\E\{R^2\mathbf 1(p\wedge p'<U\leq p\vee p')\}.
\]
The measure on the right is nonatomic, so this expression tends to zero uniformly as $|p-p'|\to0$. Asymptotic equicontinuity then gives continuous versions of the limiting processes. For the uniform laws of large numbers, we instead partition according to the finite measures $\E\{|R|\mathbf 1(U\in\cdot)\}$ and $\E\{R^2\mathbf 1(U\in\cdot)\}$. The ordinary law of large numbers at the endpoints of finitely many $L_1$ brackets proves the required uniform convergence. Applying the same argument to a finite union of classes establishes the joint statements.
\end{proof}

\begin{lemma}[Bias and Rounding]\label{lem:mean}
 Under Assumptions~\ref{asm:randomsample}--\ref{asm:moments},
 \[
 \sqrt n\sup_{p\in[0,1]}|\E\widehat A_n(p)-A(p)|\longrightarrow0.
 \]
 Linear interpolation changes $\sqrt n\widehat A_n$ by $o_{\P}(1)$ uniformly on $[0,1]$.
\end{lemma}
\begin{proof}
We first bound the bias. Define the empirical distribution of the ranks and the observed rank fraction by
\[
 H_n(p)=\frac1n\sum_{i=1}^n h_p(S_i,Y_i),\qquad
 \rho_n(p)=\frac{\lfloor np\rfloor}{n},
\]
and let $\Delta_n=\|H_n-\operatorname{id}\|_\infty+n^{-1}$. By Assumption~\ref{asm:continuity}, $M=\|\lambda\|_\infty$ is finite and $\lambda$ has a bounded modulus of continuity $\omega$. Conditional on the ranks, the difference between the cumulative conditional effect of the selected $\lfloor np\rfloor$ observations and that of the observations with $U_i\leq p$ is
\[
 \lambda(p)\{\rho_n(p)-H_n(p)\}+R_n(p),\qquad
 \|R_n\|_\infty\leq\Delta_n\omega(\Delta_n).
\]
Indeed, each observation belonging to only one of these two groups has a population rank within $\Delta_n$ of $p$. Since $\E H_n(p)=p$, taking expectations gives
\[
 \sup_{p\in[0,1]}|\E\widehat A_n(p)-A(p)|
 \leq\frac Mn+\E\{\Delta_n\omega(\Delta_n)\}.
\]
The Dvoretzky--Kiefer--Wolfowitz--Massart inequality \citep{massart1990comparison} implies $\E(n\Delta_n^2)=O(1)$. Because $\Delta_n\to_{\P}0$ and $\omega$ is bounded, Cauchy--Schwarz gives
\[
 \sqrt n\E\{\Delta_n\omega(\Delta_n)\}
 \leq\{\E(n\Delta_n^2)\}^{1/2}
       \{\E\omega(\Delta_n)^2\}^{1/2}\longrightarrow0.
\]
This proves the bias assertion.

We next consider linear interpolation between the rank knots. Its error is at most $n^{-1}\max_i|\widehat\psi_i|$, where $\widehat\psi_i=T_iY_i/\pi_n-(1-T_i)Y_i/q_n$ is the pseudo-outcome before sorting. For any square-integrable mark $R$ and $\epsilon>0$,
\[
 n\P(|R|>\epsilon\sqrt n)
 \leq\epsilon^{-2}\E\{R^2\mathbf 1(|R|>\epsilon\sqrt n)\}
 \longrightarrow0.
\]
Applying this bound to each potential outcome and using the fact that $\pi_n$ and $q_n$ are bounded away from zero gives $\max_i|\widehat\psi_i|/\sqrt n=o_{\P}(1)$. The interpolation assertion follows.
\end{proof}

We now use these lemmas to prove Proposition~\ref{prop:donsker}.

\begin{proof}
Define the empirical arm processes with the normalization used in Section~\ref{subsec:inference}:
\[
 \mathbb G_{tn}(h)=\sqrt n(\mathbb P_{tn}h-P_th),\qquad t\in\{0,1\}.
\]
By Lemmas~\ref{lem:arm_samples} and~\ref{lem:marked_thresholds}, these processes converge jointly to the independent processes $\mathbb G_t$ in Equation~\eqref{eq:arm_gaussian_process}. In particular, the empirical rank process satisfies
\[
 B_n(p)=\sqrt n\{H_n(p)-p\}
       =\pi_n\mathbb G_{1n}(h_p)+q_n\mathbb G_{0n}(h_p)
       \rightsquigarrow B_{\rm r}(p).
\]
Its limiting covariance is $p\wedge p'-pp'$, as follows by substituting the arm covariances into the definition of $B_{\rm r}$.

We next account for the empirical subgroup threshold. For $0<\lfloor np\rfloor<n$, let $\widehat u_n(p)$ be the $\lfloor np\rfloor$th smallest value of $U_i$. Set $\widehat u_n(p)=0$ when $\lfloor np\rfloor=0$ and $\widehat u_n(1)=1$. Then, almost surely,
\[
 H_n\{\widehat u_n(p)\}=\rho_n(p),\qquad
 \|\widehat u_n-\operatorname{id}\|_\infty=O_{\P}(n^{-1/2}).
\]
The defining identity can be written as
\[
 \sqrt n\{\widehat u_n(p)-p\}
 =-B_n\{\widehat u_n(p)\}+\sqrt n\{\rho_n(p)-p\}.
\]
Since $\sup_p|\rho_n(p)-p|\leq n^{-1}$, asymptotic equicontinuity of $B_n$ yields
\begin{equation}
 \sup_{p\in[0,1]}\left|\sqrt n\{\widehat u_n(p)-p\}+B_n(p)\right|
 =o_{\P}(1).
 \label{eq:quantile_expansion_appendix}
\end{equation}
The endpoint conventions ensure that the same expansion holds at zero and one.

By the definition of the ordered estimator, we have the exact identity
\[
 \widehat A_n(p)
 =\mathbb P_{1n}\{y h_{\widehat u_n(p)}(s,y)\}
  -\mathbb P_{0n}\{y h_{\widehat u_n(p)}(s,y)\}.
\]
Lemma~\ref{lem:marked_thresholds} permits replacement of $\widehat u_n(p)$ by $p$ in the centered empirical processes, uniformly in $p$. For the population term, the identity $A(p)=\int_0^p\lambda(u)\,\d u$ and uniform continuity of $\lambda$ imply
\[
 A\{\widehat u_n(p)\}-A(p)
 =\lambda(p)\{\widehat u_n(p)-p\}+o_{\P}(n^{-1/2})
\]
uniformly on $[0,1]$. Combining these expansions with Equation~\eqref{eq:quantile_expansion_appendix} gives
\[
 \sqrt n(\widehat A_n-A)(p)
 =\mathbb G_{1n}\{y h_p(s,y)\}
  -\mathbb G_{0n}\{y h_p(s,y)\}
  -\lambda(p)B_n(p)+o_{\P}(1).
\]
Joint convergence of the arm processes proves Equation~\eqref{eq:cumulative_fclt}. The limiting paths are continuous by Lemma~\ref{lem:marked_thresholds} and continuity of $\lambda$. Lemma~\ref{lem:mean} shows that the same limit holds under linear interpolation.

We also record the covariance kernel for use in the remaining proofs. For $t\in\{0,1\}$, define
\begin{equation}
 \begin{aligned}
 \phi_{1,p}(s,y)&=\frac{y h_p(s,y)-\mu_1(p)}\pi
                    -\lambda(p)\{h_p(s,y)-p\},\\
 \phi_{0,p}(s,y)&=-\frac{y h_p(s,y)-\mu_0(p)}q
                    -\lambda(p)\{h_p(s,y)-p\}.
 \end{aligned}
 \label{eq:complete_influence_functions}
\end{equation}
These functions are centered under their respective arm laws. Independence of the arms and Equation~\eqref{eq:arm_gaussian_process} therefore give
\begin{equation}
 \begin{split}
 K(p,p')=\Cov\{G(p),G(p')\}
 &=\pi\E\{\phi_{1,p}(S,Y(1))\phi_{1,p'}(S,Y(1))\}\\
 &\quad+q\E\{\phi_{0,p}(S,Y(0))\phi_{0,p'}(S,Y(0))\}.
 \end{split}
 \label{eq:complete_kernel}
\end{equation}
\end{proof}

\section{Proof of Proposition~\ref{prop:representation}}
\label{app:representation}

\begin{proof}
We construct the representation by first allowing the treatment count to vary and then removing its contribution. Introduce an auxiliary $T^B\sim\operatorname{Bernoulli}(\pi)$,
independent of the sampled unit, and define
\[
 Z=\frac{T^B Y(1)}\pi-\frac{(1-T^B)Y(0)}q,\qquad
 \eta=\frac{T^B-\pi}{\sqrt{\pi q}}.
\]
Then $\E(Z\mid U)=\lambda(U)$, $\E(\eta\mid U)=0$, and
$\E\eta^2=1$.  In particular,
\begin{align*}
 \E[\{Z-\lambda(U)\}^2h_p(S,Z)]&=r(p),\\
 \E[\{Z-\lambda(U)\}\eta h_p(S,Z)]
  &=\frac{q\mu_1(p)+\pi\mu_0(p)}{\sqrt{\pi q}}=d(p).
\end{align*}
The first identity expresses $r(p)$ as a cumulative nonnegative second moment, so $r$ is nondecreasing. Assumption~\ref{asm:moments} and continuity of the rank distribution also imply that both $r$ and $d$ are continuous.

The centered influence function for the rank-based estimator in this
auxiliary population is
\begin{equation}
 X_p=\{Z-\lambda(p)\}h_p(S,Z)-A(p)+p\lambda(p).
 \label{eq:bernoulli_influence}
\end{equation}
We separate this influence function into a conditional residual and a centered function of the rank:
\[
 X_p=\{Z-\lambda(U)\}h_p(S,Z)+f_p(U),\qquad
 f_p(U)=\{\lambda(U)-\lambda(p)\}\mathbf 1(U\leq p)-A(p)+p\lambda(p).
\]
The residual has conditional mean zero given $U$ and is orthogonal
to every function of $U$.  The Gaussian process generated by these
residuals has covariance $r(p\wedge p')$ and hence the law of
$M(p)=W\{r(p)\}$.  Its joint Gaussian limit with $\eta$ defines a
standard normal $N$ with $\Cov\{M(p),N\}=d(p)$.
The Gaussian rank process generated by $f_p(U)$ is independent of
$(M,N)$ and has the representation
\[
 H(p)=\int_0^p\{\lambda(u)-\lambda(p)\}\,\d B(u).
\]
For a deterministic square-integrable $g$, this integral is defined
using $B(p)=W_U(p)-pW_U(1)$ as
\begin{equation}
 \int_0^p g(u)\,\d B(u)
  =\int_0^p g(u)\,\d W_U(u)-W_U(1)\int_0^p g(u)\,\d u.
 \label{eq:bridge_integral}
\end{equation}
By construction, $H$ has the covariance of the centered functions $f_p$. Independence of $H$ and $M$ therefore implies that $G^+=M+H$ has covariance $\E(X_pX_{p'})$.
Its paths are continuous because $\lambda$ is continuous and
$H(p)=\int_0^p\lambda(u)\,\d B(u)-\lambda(p)B(p)$.

We now remove the treatment-count component.  Direct calculation gives
\[
 \E(X_p\mid T^B=1)=\frac{q\mu_1(p)+\pi\mu_0(p)}\pi,
 \qquad
 \E(X_p\mid T^B=0)=-\frac{q\mu_1(p)+\pi\mu_0(p)}q.
\]
These two conditional means equal $d(p)\eta$.  Consequently,
$X_p-d(p)\eta$ equals $\phi_{1,p}(S,Y(1))$ or $\phi_{0,p}(S,Y(0))$ in
Equation~\eqref{eq:complete_influence_functions}, according to the
assigned arm.  It follows that
\begin{align}
 \Cov\{G^+(p),G^+(p')\}&=K(p,p')+d(p)d(p'),\nonumber\\
 \Cov\{G^+(p)-d(p)N,G^+(p')-d(p')N\}&=K(p,p'),
 \label{eq:design_covariance_identity}
\end{align}
and $\Cov\{G^+(p)-d(p)N,N\}=0$. Since centered Gaussian processes are determined by their covariance kernels, the second identity proves Equation~\eqref{eq:design_representation}. The zero cross covariance also establishes independence from $N$, completing the proof.
\end{proof}

The same decomposition also gives a representation with separate Brownian bridges for treatment assignment and score ranking. To derive it, let $m_t(u)=\E\{Y(t)\mid U=u\}$,
$\sigma_t^2(u)=\V\{Y(t)\mid U=u\}$, and put
\[
 Q(p)=\int_0^p\left\{\frac{\sigma_1^2(u)}\pi+
                            \frac{\sigma_0^2(u)}q\right\}\,\d u,
 \qquad a(u)=\frac{qm_1(u)+\pi m_0(u)}{\sqrt{\pi q}}.
\]
Then $d(p)=\int_0^p a(u)\,\d u$ and
$r(p)=Q(p)+\int_0^p a(u)^2\,\d u$.  For independent standard
processes $W_e,B_T,B_U$,
\begin{equation}
 G(p)\ \overset d=\ W_e\{Q(p)\}
   +\int_0^p a(u)\,\d B_T(u)
   +\int_0^p\{\lambda(u)-\lambda(p)\}\,\d B_U(u).
 \label{eq:three_process_representation}
\end{equation}
To verify Equation~\eqref{eq:three_process_representation}, we condition the auxiliary residual on $(U,T^B)$.
Its conditional mean is $a(U)\eta$, and the remaining residual has
cumulative variance $Q$.  Their Gaussian limits are independent.
For the allocation term, the covariance of $\int_0^p a(u)\,\d W_T(u)$ with $W_T(1)$ is $d(p)$. Subtracting $d(p)W_T(1)$ therefore replaces this Wiener integral by $\int_0^p a(u)\,\d B_T(u)$.  This proves Equation~\eqref{eq:three_process_representation}.
The implementation in the main text estimates cumulative moments
rather than the separate conditional variance functions in this
alternative expression.

\section{Consistency of the Estimated Reference}
\label{app:reference_consistency}

We first specify the estimated reference used in the numerical calibration. Write $p_k=k/n$, $\pi_n=n_1/n$, and $q_n=1-\pi_n$. Partition the ordered sample into $J_n$ consecutive blocks with endpoints $k_j=\lfloor nj/J_n\rfloor$, where $k_0=0$ and $k_{J_n}=n$, and let $b_j=k_j-k_{j-1}$. For $k_{j-1}<i\leq k_j$, define
\begin{equation}
 \ell_i=\frac{1}{b_j}\sum_{h=k_{j-1}+1}^{k_j}\widehat\psi_{[n,h]},\qquad
 s_i=\sqrt{\frac{b_j}{b_j-1}}\{\widehat\psi_{[n,i]}-\ell_i\},\qquad
 z_i=\frac{T_{[n,i]}-\pi_n}{\sqrt{\pi_nq_n}}.
 \label{eq:reference_residuals}
\end{equation}
We choose blocks with $b_j\geq2$. Set $\widehat\lambda_n(p)=\ell_{k+1}$ on $[p_k,p_{k+1})$ and $\widehat\lambda_n(1)=\ell_n$. The cumulative estimators are $\widehat r_n(p)=n^{-1}\sum_{i\leq\lfloor np\rfloor}s_i^2$ and $\widehat d_n(p)=n^{-1}\sum_{i\leq\lfloor np\rfloor}s_i z_i$, and $\widehat c=\widehat A_n(1)$. The estimator $\widehat v_n$ is defined in Equation~\eqref{eq:comparison_variance_est}.

Conditional on the evaluation sample, generate independent standard normal variables $\xi_i$ and $\zeta_i$. Define the Gaussian drivers at each rank knot by
\begin{equation}
 \begin{aligned}
 X_k^*&=\frac{1}{\sqrt n}\sum_{i=1}^k s_i\xi_i,&
 N^*&=\frac{1}{\sqrt n}\sum_{i=1}^n z_i\xi_i,\\
 B_k^*&=\frac{1}{\sqrt n}\left(\sum_{i=1}^k\zeta_i-p_k\sum_{i=1}^n\zeta_i\right),&
 J_k^*&=\frac{1}{\sqrt n}\left(\sum_{i=1}^k\ell_i\zeta_i-\frac{1}{n}\sum_{i=1}^k\ell_i\sum_{i=1}^n\zeta_i\right).
 \end{aligned}
 \label{eq:reference_drivers}
\end{equation}
The estimated reference is
\begin{equation}
 \widehat G_n(p_k)=X_k^*+J_k^*-\widehat\lambda_n(p_k)B_k^*-\widehat d_n(p_k)N^*.
 \label{eq:estimated_reference}
\end{equation}
Extend it constantly on each $[p_k,p_{k+1})$, retaining its separately defined value at one. The shared $\xi_i$ driver ensures $\Cov^*(X_k^*,N^*)=\widehat d_n(p_k)$, while the independent $\zeta_i$ driver supplies the rank correction. For a covariate observed in both treatment arms, use the same residual and rank construction with its observed values as pseudo-outcomes and omit the treatment-count correction, as $d_g=0$.

\begin{lemma}[Uniform Consistency of the Reference Functions]
 \label{lem:reference_estimates}
 Under the conditions of Theorem~\ref{thm:sorteduniform},
 \[
 \widehat c\to_{\P}c,\qquad
 \|\widehat\lambda_n-\lambda\|_\infty+
 \|\widehat r_n-r\|_\infty+
 \|\widehat d_n-d\|_\infty+
 \|\widehat v_n-v_c\|_\infty\to_{\P}0.
 \]
 The same conclusions hold when the residual correction factor in
 Equation~\eqref{eq:reference_residuals} is omitted.
\end{lemma}
\begin{proof}
We first establish consistency of the estimated conditional effect. The law of large numbers in each treatment arm gives $\widehat c\to_{\P}c$. Let $\mathcal U_n$ be generated by all
population ranks.  Conditional on these ranks,
$\E(\widehat\psi_i\mid\mathcal U_n)=\lambda(U_i)$.  Decompose
\begin{equation}
 \widehat\psi_i-\lambda(U_i)
 =\{\psi_i-\lambda(U_i)\}+
 (T_i-\pi_n)\left\{\frac{Y_i(1)}{\pi_n}+\frac{Y_i(0)}{q_n}\right\}.
 \label{eq:margin_error_decomposition}
\end{equation}
Conditional on the ranks, the first terms are independent and have mean zero. Conditional on all sampled potential outcomes, the assignment terms also have mean zero. It follows that the cross covariance between the two sums vanishes.
For any rank-determined block $\mathcal B$ and fixed numbers $a_i$,
complete randomization gives
\[
 \V\left\{\sum_{i\in\mathcal B}(T_i-\pi_n)a_i\right\}
 =\pi_nq_n\left\{\frac n{n-1}\sum_{i\in\mathcal B}a_i^2
              -\frac1{n-1}\left(\sum_{i\in\mathcal B}a_i\right)^2\right\},
\]
where this variance is conditional on the sampled units.  Applying
this formula in Equation~\eqref{eq:margin_error_decomposition} yields
\begin{equation}
 \E\left[\left\{\sum_{i\in\mathcal B}
   (\widehat\psi_i-\lambda(U_i))\right\}^2\biggm|\mathcal U_n\right]
 \leq C\sum_{i\in\mathcal B}\E\{Y_i(1)^2+Y_i(0)^2\mid U_i\}.
 \label{eq:block_second_moment}
\end{equation}
Here $C$ is bounded for large $n$ because $\pi_n$ stays away from
zero and one.  If $b_{\min}=\min_j b_j\asymp n/J_n$, summing
Equation~\eqref{eq:block_second_moment} over the bins gives
\[
 \E\max_j\left|\frac1{b_j}\sum_{i=k_{j-1}+1}^{k_j}
    \{\widehat\psi_{[n,i]}-\lambda(U_{[n,i]})\}\right|^2
 \leq\frac{Cn}{b_{\min}^2}\E\{Y_i(1)^2+Y_i(0)^2\}
 =O(J_n^2/n)\to0.
\]
The population ranks in any bin differ from its nominal proportions
by at most $J_n^{-1}+2n^{-1}+\|H_n-\operatorname{id}\|_\infty$.
Uniform continuity of $\lambda$ therefore proves
$\|\widehat\lambda_n-\lambda\|_\infty=o_{\P}(1)$ and
\begin{equation}
 \max_i|\ell_i-\lambda(U_{[n,i]})|=o_{\P}(1).
 \label{eq:observation_margin_consistency}
\end{equation}
This argument includes the right-continuous bin convention at
boundary knots.

We next establish consistency of the residual moments, beginning with the uncorrected residuals $e_i=\widehat\psi_{[n,i]}-\ell_i$.
Replacing $\ell_i$ by $\lambda(U_{[n,i]})$ in the cumulative squared
residuals changes their supremum by at most
\[
 2\max_i|\ell_i-\lambda(U_{[n,i]})|\frac1n\sum_i
       |\widehat\psi_{[n,i]}-\lambda(U_{[n,i]})|
     +\max_i|\ell_i-\lambda(U_{[n,i]})|^2=o_{\P}(1).
\]
At a fixed population threshold, the remaining empirical moment is
\[
 \pi_n\mathbb P_{1n}\{(y/\pi_n-\lambda(1-F(s)))^2h_p(s,y)\}
  +q_n\mathbb P_{0n}\{(-y/q_n-\lambda(1-F(s)))^2h_p(s,y)\}.
\]
Expanding the squares and applying the uniform laws in Lemma~\ref{lem:marked_thresholds} shows that this expression converges uniformly to $r(p)$. By Equation~\eqref{eq:quantile_expansion_appendix}, replacing the population threshold by its empirical counterpart changes the limiting argument by $o_{\P}(1)$. Continuity of $r$ therefore gives consistency for the ordered residual sum.

The weights $z_i$ are uniformly bounded for large $n$.
Equation~\eqref{eq:observation_margin_consistency} also permits
replacement of $\ell_i$ by $\lambda(U_{[n,i]})$ in the cumulative
cross-moment $n^{-1}\sum_{i\leq\lfloor np\rfloor} e_i z_i$.
The two arm uniform laws give its limit
\[
 \E[\{Z-\lambda(U)\}\eta h_p(S,Z)]=d(p).
\]
The same quantile argument applies to this cross-moment because $d$ is continuous. Finally, the factors $\sqrt{b_j/(b_j-1)}$ tend uniformly to one.  Multiplying
residuals by these factors changes the cumulative first and second
moments by $o_{\P}(1)$, proving the assertions for $s_i$ as well.

Finally, at a fixed threshold the empirical centered second moment
is
\[
 \frac1{\pi_n}\mathbb P_{1n}\{(y-\pi_nc)^2h_p(s,y)\}
 +\frac1{q_n}\mathbb P_{0n}\{(y+q_nc)^2h_p(s,y)\}.
\]
Its uniform limit is
\[
 \frac1\pi\E\{(Y(1)-\pi c)^2h_p(S,Y(1))\}
 +\frac1q\E\{(Y(0)+qc)^2h_p(S,Y(0))\}
 =r(p)+\int_0^p\{\lambda(u)-c\}^2\,\d u=v_c(p).
\]
Replacing $c$ by $\widehat c$ changes this by at most
$2|\widehat c-c|n^{-1}\sum_i|\widehat\psi_i-c|+(\widehat c-c)^2$.
Replacing the threshold by its empirical quantile completes the
proof of consistency of $\widehat v_n$.
\end{proof}

\begin{lemma}[Conditional Convergence of the Reference]
 \label{lem:reference_convergence}
 Under the conditions of Theorem~\ref{thm:sorteduniform}, the
 conditional law of $\widehat G_n$ given the evaluation data converges
 weakly in probability to the law of $G$ in $\ell^\infty([0,1])$.
\end{lemma}
\begin{proof}
We prove convergence of the estimated reference by considering its Gaussian components under a common construction. Extend $X_k^*$ constantly on $[p_k,p_{k+1})$ and write $X_n^*$ for the resulting process.  Its conditional covariance is
$\widehat r_n(p\wedge p')$, so its marginal law is that of
$W\{\widehat r_n(p)\}$.  Uniform consistency of $\widehat r_n$ and
uniform continuity of Wiener paths on bounded intervals imply
conditional convergence and tightness of $X_n^*$.

Let $W_U$ be a second standard Wiener process, independent of the
residual driver, and realize
$\zeta_i/\sqrt n=W_U(p_i)-W_U(p_{i-1})$.
Using the step-function convention for $\widehat\lambda_n$, the
rank term in Equation~\eqref{eq:estimated_reference} is exactly
\[
 H_n^*(p)=\int_0^p\widehat\lambda_n(u)\,\d B(u)
                        -\widehat\lambda_n(p)B(p),
 \qquad B(p)=W_U(p)-pW_U(1).
\]
Within a rank interval, the contributions after $p_k$ cancel, so
this expression is constant there, as specified in the main text.
Put $e_n=\widehat\lambda_n-\lambda$.  Doob's inequality gives
\[
 \E^*\sup_p\left|\int_0^p e_n(u)\,\d W_U(u)\right|^2
       \leq4\int_0^1 e_n(u)^2\,\d u=o_{\P}(1).
\]
The remaining difference terms are bounded by
$|W_U(1)|\|e_n\|_\infty$ and
$\|e_n\|_\infty\sup_p|B(p)|$.  Thus $H_n^*\rightsquigarrow_{\P}H$
under this coupling.

The scalar $N^*$ is exactly standard normal conditional on the
sample, and
\[
 \Cov^*(X_n^*(p),N^*)=\widehat d_n(p),\qquad
 \Cov^*(H_n^*(p),N^*)=0.
\]
All components are conditionally Gaussian. The preceding covariance limits therefore establish finite-dimensional convergence of
$(X_n^*+H_n^*,N^*)$ to $(G^+,N)$ in the construction of
Proposition~\ref{prop:representation}.  Marginal tightness gives
joint tightness.  Uniform consistency of $\widehat d_n$ and the
continuous mapping $(g,z,d)\mapsto g-dz$ now give
\[
 \widehat G_n=X_n^*+H_n^*-\widehat d_nN^*
       \rightsquigarrow_{\P}G^+-dN\overset d=G.
\]
These assertions may be made conditional along any subsequence on
which the consistent estimated functions converge almost surely;
the subsequence criterion then gives convergence in probability
of the conditional laws.
\end{proof}

\section{Proof of Theorem~\ref{thm:sorteduniform}}
\label{app:sorteduniform}

We begin with the continuity properties of Gaussian supremum probabilities. We then use these properties to justify optimization of the boundary coefficients and prove Theorem~\ref{thm:sorteduniform}.

\begin{lemma}[Gaussian Supremum Quantiles]
 \label{lem:boundary_continuity}
 Let $X$ be a centered Gaussian random element of $C[0,1]$, and let
 $h$ be bounded, bounded away from zero, and continuous except for
 finitely many specified one-sided limits and endpoint values.  On
 a fixed interval of indices $\mathcal I$, put
 \[
 S_h(X)=\max\left\{0,\sup_{p\in\mathcal I}\frac{X(p)}{h(p)}\right\}.
 \]
 Its distribution is continuous on $(0,\infty)$.  If
 $\V\{X(p_0)\}>0$ for some $p_0\in\mathcal I$, the distribution
 is strictly increasing there, and its $(1-\alpha)$ quantile is
 finite and positive for $0<\alpha<1/2$.
\end{lemma}
\begin{proof}
We first show that the distribution of $S_h(X)$ has no positive atoms. The functional $S_h$ is continuous in the supremum norm. Its convexity and positive homogeneity allow us to study its distribution through Gaussian probabilities of scaled convex sets. For a finite set of indices, the sublevel probability is log-concave in the positive level.  Indeed, the set
$\{(x,t):x_j\leq th(p_j),\ t\geq0\}$ is convex, and integrating
the Gaussian density preserves log-concavity
\citep{prekopa1973joint}.  Singular vectors are linear images of
nondegenerate vectors on their supports.  Passing to increasing
dense grids, including the specified endpoint values, preserves
log-concavity.

Every neighborhood of zero in the support of a centered Gaussian
random element has positive probability \citep{bogachev1998joint}.
Consequently $\P\{S_h(X)\leq t\}>0$ for every $t>0$, because
this event contains a sufficiently small supremum-norm ball.
The logarithm of this probability is finite and concave on $(0,\infty)$, so it is continuous there. This proves that $S_h(X)$ has no positive atoms.

We next establish strict increase of the distribution function. If $X(p_0)$ is nondegenerate, the linear support contains a path
$x$ with $S_h(x)>0$.  Scaling this path places $S_h(x)$ in any
specified open subinterval of $(0,\infty)$.  Continuity of the
functional gives a neighborhood with values in that interval, and
this neighborhood has positive Gaussian probability.  The CDF is
strictly increasing.  Finally,
$\P\{S_h(X)=0\}\leq\P\{X(p_0)\leq0\}=1/2$, and the random
variable is almost surely finite.  Right continuity at zero proves
positivity of every upper-half quantile.  The same proof applies
to $\sup|X|/h$ and to any fixed finite maximum of such functionals.
\end{proof}

\begin{lemma}[Nondegeneracy]\label{lem:nondegenerate}
 Under the main assumptions, if $v_{A(1)}(1)>0$, the process $G$
 is nonzero at some proportion in every fixed interval $[p_l,1]$
 with $p_l<1$.  The processes $G(p)-pG(1)$ and
 $G(p)-p^2G(1)$ are nonzero at some $p\in(0,1)$.
\end{lemma}
\begin{proof}
We first examine the endpoint variance,
\[
 \V\{G(1)\}=\V\{Y(1)\}/\pi+\V\{Y(0)\}/q.
\]
If this variance is positive, the first assertion follows because every stated interval contains one.  Otherwise both
potential outcomes are constants, say $y_1,y_0$.  Then
\[
 G(p)\overset d=aB_T(p),\qquad
 a=\frac{qy_1+\pi y_0}{\sqrt{\pi q}},\qquad v_{A(1)}(1)=a^2.
\]
The assumption makes $a\ne0$, proving the first assertion also
in this case.

To prove the assertion for the endpoint-adjusted processes, take either $a(p)=p$ or $a(p)=p^2$ and suppose, for a contradiction, that
$G(p)-a(p)G(1)$ is identically zero.  The two independent arm
components in Equation~\eqref{eq:complete_kernel} imply
$\phi_{t,p}(S,Y(t))=a(p)\phi_{t,1}(S,Y(t))$ almost surely, simultaneously for
rational $0<p<1$.  On the event $U>p$, the left side is constant,
whereas the right side is a nonzero multiple of
$Y(t)-\E Y(t)$.  Thus $Y(t)$ is constant on $U>p$.  Nestedness
of these events and their union $U>0$ make $Y(t)$ constant almost
surely for each arm.  Hence $G(1)=0$, and the assumed identity
makes $G$ identically zero.  This contradicts the first assertion.
\end{proof}

The direct optimization requires a uniform argument over the two
coefficients.  In particular, the selected coefficients need not
converge to a unique pair.  The next two lemmas handle this issue.

\begin{lemma}[Gaussian Boundary Events]
 \label{lem:strict_boundary}
 Let $X$ be a nonzero centered Gaussian random element of $C[0,1]$.
 For continuous positive functions $h_0,h_1$, write
 $C_i=\{x:x(p)\leq h_i(p),\ \forall p\}$ and
 $h_t=(1-t)h_0+th_1$, where $0<t<1$.  If the two events have
 positive probability and differ with positive probability, then
 \[
 \P\{X\leq h_t\}>
       \P(X\in C_0)^{1-t}\P(X\in C_1)^t.
 \]
\end{lemma}
\begin{proof}
We first give a finite-dimensional inequality and then pass to the
path space.  Let $\gamma_m$ be the standard Gaussian law on
$\mathbb R^m$, let $A_0,A_1$ have positive probabilities $p_0,p_1$,
and let a measurable set $A_t$ contain
$(1-t)A_0+tA_1$.  Put $\nu_i=\gamma_m(\cdot\mid A_i)$.
For probability measures with finite second moments, let
$W_2(\nu_0,\nu_1)$ denote their quadratic Wasserstein distance.
The displacement interpolation between $\nu_0$ and $\nu_1$ is
supported on $A_t$.  Convexity of the entropy
$\int f\log f$ along this interpolation
\citep[][Theorem~2.2]{mccann1997direct}, together with the identity
for the quadratic Gaussian potential, gives
\begin{align*}
 D(\nu_t\Vert\gamma_m)
 &\leq (1-t)D(\nu_0\Vert\gamma_m)
        +tD(\nu_1\Vert\gamma_m)
        -\frac{t(1-t)}2W_2(\nu_0,\nu_1)^2.
\end{align*}
Here $D$ denotes relative entropy. To obtain the displayed inequality, note that along an optimal transport
coupling $(Z_0,Z_1)$ the quadratic-potential contribution obeys
\[
 \E\|(1-t)Z_0+tZ_1\|^2
 =(1-t)\E\|Z_0\|^2+t\E\|Z_1\|^2
                   -t(1-t)\E\|Z_0-Z_1\|^2,
\]
while the ordinary entropy contribution is convex.  The conditional
Gaussian laws have finite entropy and second moments, so these
statements apply; truncation also gives the same inequalities for
unbounded sets.  Since
$D(\nu_i\Vert\gamma_m)=-\log p_i$ and
$D(\nu_t\Vert\gamma_m)\geq-\log\gamma_m(A_t)$, we obtain
\begin{equation}
 \log\gamma_m(A_t)\geq(1-t)\log p_0+t\log p_1
                      +\frac{t(1-t)}2W_2(\nu_0,\nu_1)^2.
 \label{eq:strict_gaussian_finite}
\end{equation}

View $X$ as a Gaussian random element of $H=L_2[0,1]$, and let
$s^2>0$ be the operator norm of its covariance on $H$.  Choose
nested finite grids with dense union, and let $C_{i,m}$ be the
corresponding boundary events.  Gaussian regression gives
\[
 X=L_mZ_m+R_m,
\]
where $Z_m$ is a standard Gaussian vector spanning the grid
evaluations, $R_m$ is independent of $Z_m$, and
$\|L_m\|^2\leq s^2$.  Singular grid covariance matrices are handled
by using only their nondegenerate span.  The bound on $L_m$ follows
because the covariance of the regression is no greater than that
of $X$ on $H$.

Let $\nu_{i,m}$ be the conditional law of $X$ given $C_{i,m}$,
and let $\widetilde\nu_{i,m}$ be that of $Z_m$ on the same event.
Coupling the two conditional vectors and using the same independent
residual $R_m$ shows
\[
 W_2(\nu_{0,m},\nu_{1,m})
       \leq s W_2(\widetilde\nu_{0,m},\widetilde\nu_{1,m}).
\]
The latent-vector boundary sets satisfy the Minkowski containment
used in Equation~\eqref{eq:strict_gaussian_finite}.  Thus
\begin{align*}
 \log\P(C_{t,m})\geq{}&(1-t)\log\P(C_{0,m})
                       +t\log\P(C_{1,m})\\
 &+\frac{t(1-t)}{2s^2}W_2(\nu_{0,m},\nu_{1,m})^2.
\end{align*}
Continuity of the paths and boundaries makes these events decrease
to $C_i$.  Their conditional densities, relative to the law of $X$,
converge in $L_1$ and are bounded by a common constant.  Together
with $\E\|X\|_H^2<\infty$, this gives convergence of the conditional
laws in $W_2$.  Taking limits yields the same inequality for the
full boundary events.  If these events differ with positive
probability, their normalized conditional laws on $H$ are different.
Indeed, continuous paths are determined by their $L_2$ classes.
Their $W_2$ distance is then positive, proving the strict inequality.
\end{proof}

\begin{lemma}[Consistent Optimization of Gaussian Boundaries]
 \label{lem:optimized_boundary}
 Let $X$ be a nonzero centered Gaussian random element of $C[0,1]$.
 Let $f_0$ be continuous and bounded away from zero, and let $f_1$
 be continuous and nonnegative.  For
 $\beta\in\mathcal B_\epsilon$, put
 $h_\beta=\beta_0f_0+\beta_1f_1$ and
 $F(\beta)=\P(X\leq h_\beta)$.  Fix $q_\alpha=1-\alpha>1/2$
 and a linear objective $J(\beta)=c_0\beta_0+c_1\beta_1$ with
 $c_0>0$ and $c_1\geq0$.

 Suppose $Z_n\rightsquigarrow X$, an estimated Gaussian reference
 $\widehat X_n\rightsquigarrow_{\P}X$ conditionally on the sample,
 $\|\widehat f_j-f_j\|_\infty\to_{\P}0$, and
 $\widehat c_j\to_{\P}c_j$.  Assume the estimated bases are
 nonnegative, with $\widehat f_0$ bounded away from zero with
 probability tending to one.  Let $\widehat\beta_n$ minimize the
 estimated linear objective subject to
 \[
 \P^*(\widehat X_n\leq
         \widehat\beta_{0n}\widehat f_0+
         \widehat\beta_{1n}\widehat f_1)\geq q_\alpha,
 \qquad \widehat\beta_n\in\mathcal B_\epsilon.
 \]
 Then the population minimum exists, the estimated objective value
 converges to it, and
 \begin{equation}
 \P\{Z_n\leq\widehat\beta_{0n}\widehat f_0+
                   \widehat\beta_{1n}\widehat f_1\}
                       \longrightarrow q_\alpha.
 \label{eq:optimized_coverage}
 \end{equation}
 The conclusion also holds for near-minimizers whose objective error
 and feasibility error tend to zero in probability.  All inequalities
 between paths are simultaneous on $[0,1]$.
\end{lemma}
\begin{proof}
We first establish existence of a population optimum. A sufficiently large constant multiple of $f_0$ is feasible. Since $c_0>0$ and $\beta_0\geq\epsilon\beta_1$, every objective sublevel set in the admissible cone is compact.
The origin is infeasible: at an index with positive variance,
$\P\{X(p)\leq0\}=1/2<q_\alpha$.  Upper semicontinuity of closed
boundary probabilities makes the feasible set closed.  A population
minimum therefore exists and its minimizing set $\mathcal O$ is
nonempty, compact, and separated from the origin.

For nonzero $\beta$ in the cone, the boundary $h_\beta$ is uniformly
positive.  By Lemma~\ref{lem:boundary_continuity}, the normalized
supremum has no atom at one and its CDF is strictly increasing at
positive levels.  Hence $F$ is continuous away
from the origin, and $F(a\beta)$ is strictly increasing in $a>0$.
Every minimizer satisfies $F(\beta)=q_\alpha$: otherwise it could
be scaled down within the cone to lower the objective.

We next show that all population minimizers define the same
acceptance event.  For $\beta,\beta'\in\mathcal O$, their midpoint
has the same objective value.  If the two boundary events differed
with positive probability, Lemma~\ref{lem:strict_boundary} would
give a midpoint probability strictly greater than $q_\alpha$.
Scaling this boundary down slightly would preserve feasibility and
lower its objective, a contradiction.  Thus all the optimal events
agree almost surely.

This agreement also holds as equality of sets on the linear support
$\mathcal L\subset C[0,1]$ of $X$.  Each optimal event is closed
and contains a relative neighborhood of zero.  If a support path
belonged to one such set but not another, a slight scaling toward
zero would put it in the relative interior of the first set while
keeping it outside the second.  A relative open neighborhood would
then have positive Gaussian probability, contradicting almost-sure
agreement.  Moreover, for an optimal boundary, a support path $x$
with $\sup_p\{x(p)-h_\beta(p)\}=0$ lies on the relative boundary
of this common set: the maximum is attained, and scaling $x$
away from zero violates the active positive boundary.  This relative
boundary has probability zero by Lemma~\ref{lem:boundary_continuity}.
Thus the coverage event is the same for every population minimizer. We use this fact below to allow the estimated coefficients to depend on the sample that determines $Z_n$.

We next prove that the estimated optimization approaches this population problem. On every compact
subset of the cone separated from the origin,
\begin{equation}
 \sup_\beta\left|
 \P^*(\widehat X_n\leq
          \beta_0\widehat f_0+\beta_1\widehat f_1)-F(\beta)
          \right|\to_{\P}0.
 \label{eq:uniform_boundary_probability}
\end{equation}
To prove Equation~\eqref{eq:uniform_boundary_probability}, observe that the boundary maps are uniformly Lipschitz in $\beta$ on the stated compact sets and have a common positive lower bound.
Conditional weak convergence gives convergence at each fixed
boundary.  A finite parameter net and the absence of boundary atoms
extend it uniformly.  The needed uniform absence of atoms follows
by compactness: along any convergent sequence of coefficients, the
supremum excess converges to a variable with no atom at zero.
Uniform consistency of the estimated bases gives the same assertion
for the fitted boundaries.

The estimated constant-boundary quantile is bounded in probability,
so it gives a feasible candidate with bounded objective.  The
minimizers are consequently bounded in probability.  They are also
separated from the origin with probability tending to one: at a
fixed nondegenerate index, a vanishing boundary has limiting coverage
$1/2$, so it cannot satisfy a constraint tending to $q_\alpha$.
Equation~\eqref{eq:uniform_boundary_probability} therefore applies
to the selected coefficients.  Every subsequential limit is
population-feasible.  Conversely, if $\beta\in\mathcal O$, then
$(1+\delta)\beta$ is strictly feasible for any $\delta>0$.
It is estimated-feasible with probability tending to one.
Optimality of $\widehat\beta_n$, followed by $\delta\downarrow0$,
shows
\begin{equation}
 \operatorname{dist}(\widehat\beta_n,\mathcal O)\to_{\P}0,
 \qquad \widehat J_n(\widehat\beta_n)
                     \to_{\P}\min_{F(\beta)\geq q_\alpha}J(\beta).
 \label{eq:optimizer_consistency}
\end{equation}
The same argument allows vanishing feasibility and objective errors.

Finally, take any subsequence and then a further joint weak limit
$(Z_n,\widehat\beta_n)\rightsquigarrow(X,\beta_\infty)$.
Equation~\eqref{eq:optimizer_consistency} implies
$\beta_\infty\in\mathcal O$ almost surely; it need not be independent
of $X$.  All these coefficients define the same closed event on
$\mathcal L$, with a common probability-zero boundary.  Hence
\[
 \P\{X\leq h_{\beta_\infty}\}=q_\alpha,
 \qquad
 \P\left[\sup_p\{X(p)-h_{\beta_\infty}(p)\}=0\right]=0.
\]
The continuous mapping theorem, uniform basis consistency, and
convergence at this continuity set prove
Equation~\eqref{eq:optimized_coverage}.  Since every subsequence
has this conclusion, it holds for the original sequence.
\end{proof}

We now apply these results to the minimum-area band.

\begin{proof}
We use Lemma~\ref{lem:optimized_boundary} with
\[
 Z_n=\sqrt n(\widehat A_n-A),\qquad X=G,\qquad
 f_0(p)=\sqrt{v_c(1)},\qquad f_1(p)=\sqrt{v_c(p)}.
\]
Proposition~\ref{prop:donsker} gives the first weak convergence,
and Lemma~\ref{lem:reference_convergence} gives the conditional
reference convergence.  Lemma~\ref{lem:reference_estimates} gives
uniform consistency of both bases.  For the objective in
Equation~\eqref{eq:CIoptim}, the coefficients are
\[
 c_0=\sqrt{v_c(1)}>0,\qquad
 c_1=\int_0^1\sqrt{v_c(p)}\,\d p,
\]
and their empirical versions converge in probability.  The process
is nontrivial by Lemma~\ref{lem:nondegenerate}.  Lemma~\ref{lem:optimized_boundary} therefore gives the cumulative coverage and minimum-area conclusions.
Dividing each inequality by $p>0$ gives
Equation~\eqref{eq:uniform_coverage} on the entire asserted range.

For Equation~\eqref{eq:minimump}, the coefficients are instead
$c_0=\sqrt{v_c(1)}\log(1/p_l)$ and
$c_1=\int_{p_l}^1\sqrt{v_c(p)}\,\d p/p$.  They are finite and
$c_0>0$ for $0<p_l<1$.  For the reference-time criterion they are
$c_0=1-t_l$ and $c_1=2(1-t_l^{3/2})/3$.  The same lemma applies.
An upper band uses $-G$.  For a two-sided band, use the sets
$\{|x(p)|\leq h_\beta(p),\ \forall p\}$ throughout; they satisfy
the same Minkowski containment and boundary-continuity argument.
If a finite sample has $\widehat v_n(1)=0$, define an uninformative
band on that event, whose probability tends to zero.
\end{proof}

\section{Effect of the Treatment-count Correction}
\label{app:count_ordering}

The representation also quantifies the conservatism from omitting
the treatment-count correction.  Let $G^+=W(r)+H$ in the proof
of Proposition~\ref{prop:representation}.  On the same probability
space, $G=G^+-dN$ and $G$ is independent of $N$.  For any fixed
continuous function $a$, define $G_a=G^+-aN$.  Then
\begin{equation}
 \Cov\{G_a(p),G_a(p')\}-K(p,p')
       =\{d(p)-a(p)\}\{d(p')-a(p')\}.
 \label{eq:projection_order}
\end{equation}
Thus the exact correction $a=d$ has the smallest covariance in
positive-semidefinite order within this family.

For completeness, \citet[][Corollary~2.2]{fogarty2026comparison}
states that if centered Gaussian vectors $X,Y$ satisfy
$\Sigma_Y-\Sigma_X\succeq0$, then for every closed convex set $C$,
\[
 \P(Y\in C)\geq\tfrac12
       \quad\Longrightarrow\quad
 \P(X\in C)\geq\P(Y\in C).
\]
We apply this result to a fixed positive boundary shape $b(p)$ and the finite-grid event $x_j\leq k b(p_j)$ for every $j$.
Increasing dense grids and continuity of the paths give the same
ordering for the supremum event.  It follows that the
$(1-\alpha)$ supremum quantile of $G/b$ is no larger than that of
$G_a/b$ for $0<\alpha<1/2$.

The ordering also holds conditionally for the estimated references
using the same residuals and shape.  Indeed,
$\widehat G_n^+=X_n^*+H_n^*$ has covariance
$\widehat d_n(p)$ with $N^*$ and $\V^*(N^*)=1$.
Subtracting $\widehat d_nN^*$ removes exactly the outer product
$\widehat d_n\widehat d_n^{\mathsf T}$ from its covariance.
This comparison concerns exact conditional quantiles; separately
simulated quantiles have additional Monte Carlo variation.
For any common admissible family and area objective, every boundary
feasible for the uncorrected reference is also feasible for the
corrected reference.  The corrected minimum area is therefore no
larger, both at the population level and for exact fitted Gaussian
probabilities.  This is an area comparison within the specified
family, not a pointwise ordering of differently optimized boundaries.
The coverage proof uses equality of the limiting laws, rather than
this ordering.

A simple case illustrates the size of the correction.  If
$\pi=1/2$ and $Y(1)=Y(0)=1/2$, then $\lambda=0$,
$r(p)=p$, and $d(p)=p$.  The uncorrected process is $W(p)$,
whereas the actual process is $W(p)-pW(1)$.

\section{Direct Boundary Optimization and Numerical Calibration}
\label{app:calibration}

We first describe the role of the boundary constraint $\beta_0\geq\epsilon\beta_1$. This constraint makes the admissible cone closed and ensures that the boundaries have a common positive lower bound on compact coefficient sets away from zero.  A fixed small $\epsilon$, such as
$10^{-3}$, permits a wide range of intercept-to-slope ratios.  A
smaller prespecified value can be used with the same proof.  For a
general corrected Gaussian process, the Wiener law of the iterated
logarithm alone does not justify excluding an optimal zero
intercept; this is why the full-range formulation specifies the
admissible family explicitly.

We compute the boundary coefficients by profiling the intercept at each candidate slope. Write $\widehat t_n(p)=\widehat v_n(p)/\widehat v_n(1)$ and
$\widehat X_n(p)=\widehat G_n(p)/\sqrt{\widehat v_n(1)}$.
For each slope $b\geq0$, the smallest feasible intercept is
\begin{equation}
 a_n(b)=\max\left\{\epsilon b,\
 Q^*_{1-\alpha}\left[
     \sup_{p\in[0,1]}\{\widehat X_n(p)-b\sqrt{\widehat t_n(p)}\}
                         \right]\right\}.
 \label{eq:profile_intercept}
\end{equation}
The maximum also enforces nonnegativity.  Thus
Equation~\eqref{eq:CIoptim} is a one-dimensional search minimizing
\[
 a_n(b)+b\widehat m_n,\qquad
 \widehat m_n=\int_0^1\sqrt{\widehat t_n(p)}\,\d p.
\]
The resulting minimizer directly specifies the coefficients in Equation~\eqref{eq:CIoptim}. To bound the range of slopes that must be considered, let $u_n=a_n(0)$.  Any slope
that improves on the constant boundary must satisfy
$\epsilon b\leq u_n$ and, when $\widehat m_n>0$,
$b\widehat m_n\leq u_n$.  These bounds can be used to prune the
search.  Other linear area objectives replace $\widehat m_n$ by the
corresponding ratio of objective weights.

For the GATES area criterion in Equation~\eqref{eq:minimump}, the two objective
weights, after removing the common scale, are
\[
 c_{0n}=\log(1/p_l),\qquad
 c_{1n}=\int_{p_l}^1\sqrt{\widehat t_n(p)}\,\frac{\d p}p.
\]
The reference-time weights are $1-t_l$ and
$2(1-t_l^{3/2})/3$.  In each case the probability constraint
continues to use the complete estimated Gaussian process.

We next show how to compute the profiled supremum from Gaussian draws. The main reference path and its boundary are constant on each
$[p_k,p_{k+1})$, with their separately specified values at one.
Therefore every supremum in Equation~\eqref{eq:profile_intercept}
is exactly a maximum over the $n+1$ rank knots.  Generating $L$
reference paths takes $O(nL)$ operations after sorting, and evaluating
$M$ candidate slopes by direct profiling takes $O(nLM)$ operations.
The reference paths may be reused across all candidate slopes.

\begin{algorithm}[!t]
 \spacingset{1}
 \hspace*{\algorithmicindent}\textbf{Input}: ordered pseudo-outcomes,
 treatment indicators, bins $J_n$, level $\alpha$, intercept ratio
 $\epsilon$, Gaussian draws $L$, and a deterministic slope grid
 $\mathcal S_L$ containing zero\\
 \hspace*{\algorithmicindent}\textbf{Output}: the optimized uniform lower band
 \begin{algorithmic}[1]
  \State Compute the bin means, $s_i,z_i$, $\widehat r_n$,
         $\widehat d_n$, $\widehat v_n$, and
         $u_k=\sqrt{\widehat v_n(p_k)/\widehat v_n(1)}$.
  \For{$l=1,\ldots,L$}
   \State Generate independent $\xi_i,\zeta_i\sim N(0,1)$.
   \State Compute $X_k^*,N^*,B_k^*,J_k^*$ by
          Equation~\eqref{eq:reference_drivers}.
   \State Store $x_{lk}=(X_k^*+J_k^*-
       \widehat\lambda_n(p_k)B_k^*-
       \widehat d_n(p_k)N^*)/\sqrt{\widehat v_n(1)}$.
  \EndFor
  \For{$b\in\mathcal S_L$}
   \State Compute $M_l(b)=\max_{0\leq k\leq n}(x_{lk}-bu_k)$.
   \State Let $a(b)$ be the greater of $\epsilon b$ and the empirical
          $(1-\alpha)$ quantile of $M_1(b),\ldots,M_L(b)$.
   \State Record the objective $a(b)+b\widehat m_n$.
  \EndFor
  \State Choose the slope with the smallest recorded objective,
         breaking ties deterministically, and set
         $(\widehat\beta_0^*,\widehat\beta_1^*)=(a(b),b)$.
  \State \textbf{return}
       $\widehat\Psi_n(p)-\{
       \widehat\beta_0^*\sqrt{\widehat v_n(1)}+
       \widehat\beta_1^*\sqrt{\widehat v_n(p)}\}/(p\sqrt n)$.
 \end{algorithmic}
 \caption{Direct optimization of the design-corrected Wiener boundary.}
 \label{alg:design_calibration}
\end{algorithm}

To control the simulation error after optimizing the coefficients, let $M_L=|\mathcal S_L|$.  For a fixed slope $b$, the simulated
$M_l(b)$ are conditionally iid.  The Dvoretzky--Kiefer--Wolfowitz
inequality and a union bound over the deterministic slope grid give
\begin{equation}
 \P^*\left\{\max_{b\in\mathcal S_L}\sup_x
       |\widehat F_{b,L}(x)-F_{b,n}(x)|>\delta\right\}
       \leq2M_L\exp(-2L\delta^2),
 \label{eq:uniform_mc_bound}
\end{equation}
where $F_{b,n}$ is the fitted reference CDF at slope $b$.
Because the bound holds over the entire deterministic slope grid, it also controls the simulation error at the selected pair. The same Gaussian draws may therefore be reused across slopes.

Choose expanding slope grids with mesh tending to zero and
$\log M_L/L\to0$, and let $L=L_n\to\infty$.
For example, slopes spaced by $L^{-1/4}$ up to $L^{1/4}$ have these
properties.  Equation~\eqref{eq:uniform_mc_bound}, uniform convergence
of the reference probabilities, and strict feasibility after a small
radial enlargement of a population optimum imply that the simulated
optimizer has vanishing feasibility and objective errors.  The
near-minimizer assertion of Lemma~\ref{lem:optimized_boundary}
therefore establishes the same coverage and area limits.

For an explicit conditional simulation guarantee, choose
$\tau_L\downarrow0$ with $\log(M_L/\tau_L)=o(L)$, set
\[
 \delta_L=\left\{\frac{\log(2M_L/\tau_L)}{2L}\right\}^{1/2},
\]
and use empirical quantiles at level $1-\alpha+\delta_L$.
On an event of conditional probability at least $1-\tau_L$, every
selected boundary is then feasible for the fitted reference at
level $1-\alpha$.  Use an uninformative band if the inflated level
exceeds one.  The inflation vanishes asymptotically.  These statements
control Gaussian-simulation error, not finite-sample approximation of
the estimator by its Gaussian limit.

The mean-adjusted optimization uses the same profiling algorithm
with $\sqrt{\widehat t_n(p)\{1-\widehat t_n(p)\}}$ in place of
$\sqrt{\widehat t_n(p)}$ and with the transformed reference.
Appendix~\ref{app:sortedunif_mean} gives the exact interval endpoint
calculation needed for its supremum.  Equation~\eqref{eq:uniform_mc_bound}
then applies to its profiled maxima in the same way.

\section{Proof of Corollary~\ref{cor:conf_band_k_family}}
\label{app:conf_band_k_family}

\begin{proof}
We first consider $\eta>0$ and define
$h_{\kappa,\eta}(p)=\sqrt{v_c(1)}
       \{v_c(p)/v_c(1)\vee\eta\}^{\kappa}$.
The floor $\eta$ makes this population shape bounded away from zero. Its
estimated version converges uniformly by
Lemma~\ref{lem:reference_estimates}.  The functional central limit
theorem and Lemma~\ref{lem:reference_convergence} give convergence of
both the actual and conditional reference suprema, divided by their
estimated shapes, to $\sup_pG(p)/h_{\kappa,\eta}(p)$.
By Lemmas~\ref{lem:boundary_continuity} and~\ref{lem:nondegenerate},
its upper-half quantile is positive and its CDF is continuous and
strictly increasing there.  The coefficient in
Equation~\eqref{eq:k_family} is precisely the fitted reference
quantile, so it converges to this population coefficient.  Slutsky's
theorem proves the asserted exact coverage.

If $\eta=0$ and $v_c(p_l)>0$, the same shape is bounded away from
zero on $[p_l,1]$ and the same proof applies on that interval.
For $\kappa=0$, the shape is the positive constant $\sqrt{v_c(1)}$
throughout $[0,1]$.  Division by $p>0$ completes the proof.
\end{proof}

\section{Proof of Proposition~\ref{prop:general_metric}}
\label{app:general_metric}

\begin{proof}
We apply the decomposition from Proposition~\ref{prop:representation} to the new evaluation metric. In this proof, the arm laws and empirical measures include the observed covariates together with the score and outcome, so that functions of $R(t)$ and $H(\bX)$ can be evaluated. Define the arm-specific marks by
\[
 Z_1^g=R(1)/\pi+H(\bX),\qquad Z_0^g=-R(0)/q+H(\bX),
\]
and let $Z_g=T^B Z_1^g+(1-T^B)Z_0^g$ in the auxiliary Bernoulli population.
Then $\E(Z_g\mid U)=\lambda_g(U)$, and its residual variance
is $r_g(p)=\E[\{Z_g-\lambda_g(U)\}^2h_p(S,Z_g)]$.
Writing $A_g(p)=p\Psi^g(p)$, its cumulative influence function is
\[
 X_{g,p}=\{Z_g-\lambda_g(p)\}h_p(S,Z_g)-A_g(p)+p\lambda_g(p).
\]
As in the proof of Proposition~\ref{prop:representation}, we separate this influence function into a conditional residual and a centered function of $U$.
The resulting Gaussian process has law
$W\{r_g(p)\}+\int_0^p\{\lambda_g(u)-\lambda_g(p)\}\,\d B(u)$.
Its covariance with the standardized treatment-count normal is
\[
 \E(X_{g,p}\eta)=\sqrt{\pi q}\,
      \E\{(Z_1^g-Z_0^g)h_p(S,Z_g)\}=d_g(p).
\]
The term $H(\bX)$ cancels in $Z_1^g-Z_0^g$.

Conditional on $T^B=t$, subtracting $d_g(p)\eta$ from $X_{g,p}$
gives
\[
 Z_t^g h_p(S,Z_t^g)-\E\{Z_t^g h_p(S,Z_t^g)\}-\lambda_g(p)\{h_p(S,Z_t^g)-p\}.
\]
The two-arm empirical-process argument establishes that these are
the influence functions of the actual estimator, with weights
$\pi$ and $q$ in its covariance.  At finite $n$, the corresponding
expressions use $\pi_n,q_n$. The weighted population mean of the
uncentered arm terms is exactly $A_g(p)$, so no rate beyond $\pi_n\to\pi$ is needed.
This proves the exact design-corrected representation for the new
metric.

The moment condition gives square-integrable arm marks.  Conditional
on ranks, decompose the pseudo-outcome error as
\[
 \widehat\psi_i^g-\lambda_g(U_i)
 =\{\psi_i^g-\lambda_g(U_i)\}
   +(T_i-\pi_n)\{R_i(1)/\pi_n+R_i(0)/q_n\}.
\]
Equation~\eqref{eq:block_second_moment} therefore holds with
$R_i(1)^2+R_i(0)^2+H(\bX_i)^2$ as its integrable bound.
The proof of Lemma~\ref{lem:reference_estimates} therefore gives uniform consistency of the block means and the cumulative moments used in the reference. Applying the same Gaussian construction as in Lemma~\ref{lem:reference_convergence} establishes conditional reference convergence.

To verify nondegeneracy, if the endpoint variance is zero then
$Z_1^g$ and $Z_0^g$ must each be constants, say $z_1,z_0$.
In that case the actual process is
$\sqrt{\pi q}(z_1-z_0)B_T(p)$, while
$v_g(1)=\pi q(z_1-z_0)^2$.  Otherwise its endpoint is
nondegenerate.  Thus $v_g(1)>0$ gives a nonzero process, including
on the fixed trimmed intervals of the power corollary.  Theorem~\ref{thm:sorteduniform}
and Corollary~\ref{cor:conf_band_k_family} now apply.

For Equation~\eqref{eq:proportion_dependent_cost}, on the joint
coverage event each cost mean is at most $U_j$ when its coefficient
is nonnegative and at least $L_j$ when it is negative.  Substitution
in the treatment-effect lower bound proves the displayed inequality
simultaneously in $p$.  A union bound gives the specified joint
coverage probability.
\end{proof}

\section{Proof of Proposition~\ref{prop:sortedunif_mean}}
\label{app:sortedunif_mean}

\begin{proof}
We begin with the exact relation between the original and mean-adjusted estimators. Recall that $\rho_n(p)=\lfloor np\rfloor/n$. Then
\[
 \widehat A_n^M(p)=\widehat A_n(p)-p\rho_n(p)\widehat A_n(1).
\]
Since $\sup_p|p\rho_n(p)-p^2|\leq n^{-1}$ and
$\widehat A_n(1)=O_{\P}(1)$,
\[
 \sqrt n(\widehat A_n^M-A^M)(p)
 =\sqrt n(\widehat A_n-A)(p)
       -p^2\sqrt n\{\widehat A_n(1)-A(1)\}+o_{\P}(1)
\]
uniformly.  The map $x(p)\mapsto x(p)-p^2x(1)$ is continuous
in the supremum norm.  Proposition~\ref{prop:donsker} therefore
gives the limit $G^M$, and Lemma~\ref{lem:reference_convergence}
gives conditional convergence of $\widehat G_n^M$ to the same law.

Apply Lemma~\ref{lem:optimized_boundary} to these two processes
with the continuous bases
\[
 f_0^M(p)=\sqrt{v_c(1)},\qquad
 f_1^M(p)=\sqrt{v_c(1)}\sqrt{t_c(p)\{1-t_c(p)\}},\qquad
 t_c(p)=v_c(p)/v_c(1).
\]
Uniform consistency of $\widehat v_n$ gives uniform consistency
of the bases and of the area objective.  The first basis is
strictly positive and $G^M$ is nontrivial by
Lemma~\ref{lem:nondegenerate}.  The optimized-boundary lemma yields
exact asymptotic cumulative coverage and convergence of the
minimum area.  At $p=1$ both target and estimator are exactly zero;
replacing the positive calibration boundary there by zero therefore
preserves the coverage event.  Division by $p>0$ proves
Equation~\eqref{eq:mean_uniform_coverage}.

The covariance projection is also preserved.  Put
$d^M(p)=d(p)-p^2d(1)$.  The transformed uncorrected process is
$G^{M,+}=G^M+d^M N$, with $G^M$ independent of $N$.
Its covariance exceeds that of $G^M$ by
\begin{equation}
 \{d(p)-p^2d(1)\}\{d(p')-(p')^2d(1)\}.
 \label{eq:mean_covariance_comparison}
\end{equation}
The same identity holds for the fitted references.  Accordingly,
the feasible-set and minimum-area comparison in
Appendix~\ref{app:count_ordering} applies to the mean-adjusted
family as well.
\end{proof}

We next give the interval calculation used to calibrate this transformed reference. On $[p_k,p_{k+1})$, write $g_k=\widehat G_n(p_k)$ and
$g_n=\widehat G_n(1)$.  The transformed numerator is
$g_k-p^2g_n$.  The boundary in Equation~\eqref{eq:mean_shape}
is constant on that interval, so both its excess and its ratio
are monotone in $p^2$.  The supremum is obtained by checking
$p_k$ and the left limit at $p_{k+1}$ for every interval, together
with the separate endpoint $p=1$, which contributes zero.
This gives the exact fitted supremum using $O(n)$ operations per
path and candidate slope.  The same two endpoint checks apply
to the profiled intercept calculation.  For the two-sided problem,
check both signs.

The same argument applies to other endpoint adjustments. For a fixed continuous $a$ with $a(0)=0,a(1)=1$, the map
$x(p)\mapsto x(p)-a(p)x(1)$ is continuous in the supremum norm.
Using the exact sample analogue of $A(p)-a(p)A(1)$, no rounding
term is needed.  Conditional reference convergence and consistency
of the two boundary bases follow as above.  If the transformed
limit is nontrivial, Lemma~\ref{lem:optimized_boundary} proves exact
coverage and area consistency.  Lemma~\ref{lem:nondegenerate}
supplies nontriviality for $a(p)=p$ and $a(p)=p^2$.
The endpoint remains exactly zero.  Multiplying a uniform
inequality by a fixed nonnegative weight and integrating proves
the weighted-integral statements whenever the integrals exist.
Decomposing a signed weight into its positive and negative parts
uses jointly valid lower and upper bands.

\section{Proof of Proposition~\ref{prop:donsker_uncertain}}
\label{app:donsker_uncertain}

\begin{proof}
We first bound the difference between the empirical curves under the fitted and population scoring rules. Let $\mathcal E_n$ be the event
$\sup_{\bx}|\widehat f(\bx)-f(\bx)|\leq r_n$, whose probability
tends to one.  Put $S_i=f(\bX_i)$ and
$\widehat S_i=\widehat f(\bX_i)$.
On $\mathcal E_n$, corresponding order statistics of the two score
vectors differ by at most $r_n$.  If an observation is in the top
$k$ under just one ordering, its true score is within $2r_n$ of
the $k$th highest true score.  The selections agree at $k=0,n$.
Consequently,
\begin{align}
 \sup_p|\widehat A_n^{\widehat f}(p)-\widehat A_n^f(p)|
 &\leq\sup_s\frac1n\sum_i|\widehat\psi_i|
                         \mathbf{1}(|S_i-s|\leq2r_n)\nonumber\\
 &\leq C\sum_{t=0}^1\sup_u\mathbb P_{tn}
                       \{|y|\mathbf{1}(|s-u|\leq2r_n)\}
 \quad\text{on }\mathcal E_n.
 \label{eq:score_interval_bound}
\end{align}

The classes $\{(s,y)\mapsto |y|\mathbf{1}(s\leq u):u\in\mathbb R\}$ are Donsker
by the proof of Lemma~\ref{lem:marked_thresholds}.  Because $S$
is continuously distributed, the finite measures
$\E\{Y_i(t)^2\mathbf{1}(S\in\cdot)\}$ are nonatomic.
Their mass on intervals of length $4r_n$ tends uniformly to zero.
Asymptotic equicontinuity therefore gives
\[
 \sup_u|\mathbb G_{tn}\{|y|\mathbf{1}(|s-u|\leq2r_n)\}|
       =o_{\P}(1).
\]
The population means in Equation~\eqref{eq:score_interval_bound}
are $o(n^{-1/2})$ by Equation~\eqref{eq:weighted_margin}.
This proves Equation~\eqref{eq:score_equivalence}.

The endpoint estimate $\widehat c=\widehat A_n(1)$ is identical
under both orderings.  The same deterministic argument bounds
the difference of cumulative squared pseudo-outcome sums by their
mass on shrinking score intervals.  The integrable threshold
uniform laws and nonatomic second-moment measures show that this
difference is $o_{\P}(1)$ uniformly.  Hence $\widehat v_n$ computed
with $\widehat f$ converges to $v_c^f$.

The two bin estimates have common endpoints and minimum bin size
$b_{\min}\asymp n/J_n$.  Thus
\[
 \|\widehat\lambda_n^{\widehat f}-\widehat\lambda_n^f\|_\infty
 \leq\frac{2n}{b_{\min}}
       \|\widehat A_n^{\widehat f}-\widehat A_n^f\|_\infty
 =O(J_n)o_{\P}(n^{-1/2})=o_{\P}(1).
\]
Lemma~\ref{lem:reference_estimates}, applied to the fixed rule $f$,
then gives convergence to $\lambda_f$.
To relate fitted bin values to the ranks of individual observations,
let $F$ be the population CDF of the true score.  Uniform continuity
of $F$ and the ordering discrepancy show that an observation's
empirical rank under $\widehat f$ differs from $1-F(S_i)$ by
$o_{\P}(1)$ uniformly over observations.  One can see this by
bounding the number of observations whose order relative to $i$
can change by $\#\{j:|S_j-S_i|\leq2r_n\}$, and using the
uniform empirical CDF bound.  Hence uniform continuity of
$\lambda_f$ gives
\[
 \max_i|\ell_i^{\widehat f}-\lambda_f(1-F(S_i))|=o_{\P}(1).
\]

We can therefore replace the fitted block means by these conditional means in the residual moments.  Cauchy--Schwarz gives a uniform $o_{\P}(1)$
change for squared residuals, and boundedness of $z_i$ gives the
same for residual/count cross-moments.  The remaining cumulative
sums ordered by $\widehat f$ differ from those ordered by $f$
only on shrinking score intervals.  Their absolute marks are
integrable, so the threshold uniform laws give a uniform
$o_{\P}(1)$ difference.  This proves convergence of both
$\widehat r_n$ and $\widehat d_n$ to their limits under $f$.
The within-bin residual factors converge uniformly to one and do
not change these conclusions.

The proof of Lemma~\ref{lem:reference_convergence} now gives
conditional convergence of the design-corrected reference to $G_f$.
Equation~\eqref{eq:score_equivalence} transfers the actual functional
limit to the ordering from $\widehat f$.  The estimated shapes have
the same limits as under $f$, so the optimized-boundary argument in
Lemma~\ref{lem:optimized_boundary} and the scalar argument for the
power family establish
the stated coverage results for the population-rule targets.
\end{proof}

\section{Post-selection Performance Guarantees}
\label{app:selection}

We derive the selection guarantees by evaluating the uniform inequality at the selected proportion. Let $\mathcal C_n$ be the simultaneous coverage event in
Theorem~\ref{thm:sorteduniform}, and write
$L_n(p)=\widehat\Psi_n(p)-C_n(p,\alpha)$.
On $\mathcal C_n$, every measurable data-selected proportion
$\widetilde p_n\in(0,1]$ satisfies
$\Psi(\widetilde p_n)\geq L_n(\widetilde p_n)$.
In particular, this holds for a maximizer of the lower band over
a specified candidate set, with deterministic tie breaking.
If the population maximum is attained at $p^*$, then
$\Psi(p^*)\geq\Psi(\widetilde p_n)$ as well.
Each probability is at least $\P(\mathcal C_n)\to1-\alpha$.

For a possibly data-dependent threshold $a_n$, suppose the rule
reports only proportions satisfying $L_n(\widetilde p_n)\geq a_n$.
On $\mathcal C_n$, every reported group satisfies
$\Psi(\widetilde p_n)\geq a_n$.  When no feasible proportion
exists, no group is certified.  The probability of reporting a
group whose true GATES is below its reported threshold therefore
has upper limit at most $\alpha$.  This is an unconditional
guarantee; it does not condition on finding a feasible group.
The same argument applies to the other uniform bands and to
jointly calibrated covariate descriptions.

\end{document}